\documentclass[a4paper,UKenglish,cleveref, autoref, thm-restate]{arxiv}

\nolinenumbers

\PassOptionsToPackage{color}{xy}

\usepackage{mathrsfs}
\usepackage{dsfont}
\usepackage{bbm}
\usepackage{palatino}

\usepackage[usenames,dvipsnames]{xcolor}
\usepackage{graphicx, calc}
\graphicspath{{./figures/}}
\usepackage{float}
\usepackage{bm}
\usepackage{stmaryrd}
\usepackage{stackengine}
\usepackage{cite}
\usepackage{tikz}
\usetikzlibrary{calc,intersections,through,backgrounds}
\usepackage[ruled]{algorithm}
\usepackage{algorithmicx}
\usepackage{algpseudocode}
\usepackage{listings}
\usepackage{thmtools}
\usepackage{thm-restate}
\usepackage{enumerate}

\usepackage{qcircuit}

\makeatletter
\renewcommand{\fnum@figure}{Fig. \thefigure}
\makeatother

\newcommand{\ket}[1]{{\ensuremath{\lvert#1\rangle}}}

\newlength\myheight
\newlength\mydepth
\settototalheight\myheight{Xygp}
\newcommand{\se}{\subseteq}
\newcommand{\ls}{\leqslant}
\newcommand{\gs}{\geqslant}
\newcommand{\sm}{\setminus}

\newcommand{\G}[1][q]{\mathbb{G}_{#1}}
\newcommand{\Gphi}[1][q]{\mathbb{G}_{#1\mid\phi}}
\renewcommand{\theenumi}{\roman{enumi}}

\title{Vertex-minor universal graphs for  generating  entangled quantum subsystems}

 \author{Maxime Cautrès}{Universit\'e Grenoble Alpes, CEA-L\'eti, F-38054 Grenoble, France \and École Normale Supérieure de Lyon, F-69007 Lyon, France }{maxime.cautres@ens-lyon.org}{https://orcid.org/0000-0002-7899-1866}{}

 \author{Nathan Claudet}{Inria Mocqua, LORIA, CNRS, Université de Lorraine, F-54000 Nancy, France}{nathan.claudet@inria.fr}{https://orcid.org/0009-0000-0862-1264}{}

 \author{Mehdi Mhalla}{Université Grenoble Alpes, CNRS, Grenoble INP, LIG, F-38000 Grenoble, France}{mehdi.mhalla@univ-grenoble-alpes.fr}{https://orcid.org/0000-0003-4178-5396}{}

\author{Simon Perdrix}{Inria Mocqua, LORIA, CNRS, Université de Lorraine, F-54000 Nancy, France}{simon.perdrix@loria.fr}{https://orcid.org/0000-0002-1808-2409}{}

\author{Valentin Savin}{ Universit\'e Grenoble Alpes, CEA-L\'eti, F-38054 Grenoble, France}{valentin.savin@cea.fr}{https://orcid.org/0000-0001-9362-8769}{}

\author{Stéphan Thomassé}{Université de Lyon, École Normale Supérieure de Lyon, UCBL, CNRS, LIP, F-69007 Lyon, France}{stephan.thomasse@ens-lyon.fr}{https://orcid.org/0000-0002-7090-1790}{}

\authorrunning{M. Cautrès, N. Claudet, M. Mhalla, S. Perdrix, V. Savin and S. Thomassé} 

\Copyright{Maxime Cautrès, Nathan Claudet, Mehdi Mhalla, Simon Perdrix, Valentin Savin and Stéphane Thomassé} 



\category{} 

\EventEditors{John Q. Open and Joan R. Access}
\EventNoEds{2}
\EventLongTitle{51st International Colloquium on Automata, Languages, and Programming (ICALP 2024)}
\EventShortTitle{ICALP 2024}
\EventAcronym{ICALP}
\EventYear{2024}
\EventDate{July 8--12, 2024}
\EventLocation{Talline, Estonia}
\EventLogo{}
\SeriesVolume{51}
\ArticleNo{23}

\begin{document}

\maketitle

\begin{abstract}
We study the notion of  $k$-stabilizer universal quantum state, that is, an $n$-qubit quantum state, such that it is possible to induce any stabilizer state on any $k$ qubits, by using only local operations and classical communications. These states generalize the notion of $k$-pairable states introduced by Bravyi et al., and can be studied from a combinatorial perspective using graph states and $k$-vertex-minor universal graphs. First, we demonstrate the existence of $k$-stabilizer universal graph states that are optimal in size with $n=\Theta(k^2)$ qubits. We also provide parameters for which a random graph state on $\Theta(k^2)$ qubits is $k$-stabilizer universal with high probability. Our second  contribution consists of two explicit constructions of $k$-stabilizer universal graph states on $n = O(k^4)$ qubits. Both  rely upon the incidence graph of the projective plane over a finite field $\mathbb{F}_q$. This provides a major improvement over the previously known explicit construction of $k$-pairable graph states with $n = O(2^{3k})$, 
bringing forth a new and potentially powerful family of multipartite quantum resources.
\end{abstract}

\section{Introduction}

Quantum communication networks often rely on classical communication along with pre-shared entanglement. 
In this context, a highly pertinent problem is to explore which resource states enable a group of $n$ parties, equipped with the capability of employing Local  Operations and Classical Communication (LOCC), to create entangled EPR pairs among any $k$ pairs of qubits. It is only recently that Bravyi et al.~addressed this fundamental question and provided both upper and lower bounds for what they called the $k$-pairability of quantum states, in terms of the number of parties and the number of qubits per party needed for a quantum state to be $k$-pairable ~\cite{bravyi2022generating}. Formally, an $n$-party state $\ket \psi$ is said to be \textit{$k$-pairable} if, for every $k$ disjoint pairs of parties  $\{a_1, b_1\},\ldots,\{a_k, b_k\}$, there exists a LOCC protocol that starts with $\ket \psi$ and ends up with a state where each of those $k$ pairs of parties shares an EPR pair. Bravyi et al.~studied $n$-party states in the case where each party holds  $m$ qubits, with $m$ ranging from $1$ to $\log(n)$. In the case where each party holds at least $m = 10$ qubits, they showed the existence of $k$-pairable states where $k$ is of the order of $n/ \text{polylog}(n)$, which is nearly optimal when $m$ is constant. They also showed that if one allows a logarithmic number of qubits per party, then there exist $k$-pairable states with $k = n/ 2$. Moreover, before their work, numerous variations of this problem had surfaced in the literature, some in the context of entanglement routing~\cite{schoute2016shortcuts,hahn2019quantum,pant2019routing}, and some about problems that can be described as variants of $k$-pairability, for example (i) preparing resource states by clustering and merging~\cite{Miguel_Ramiro_2023}, or  starting from a particular network~\cite{DSL:multipoint,Contreras_Tejada_2022}, (ii) creating one EPR pair hidden from other parties~\cite{illianoetal}, (iii) studying the complexity and providing algorithms for generating EPR pairs within a predetermined set~\cite{dahlberg2020transforming,DHW:howtotransform}, (iv) estimating the cost of distributing a graph state in terms of EPR pairs~\cite{meignant2019distributing,fischer2021distributing}, and (v) transforming an arbitrary entangled state into another through a network where each party holds multiple qubits \cite{christandl2023resource}. 

The notion of $k$-pairability that we focus on in the present paper relates to the scenario that is both the most natural and challenging~\cite{bravyi2022generating}, when each party possesses precisely one qubit, \emph{i.e.},~$m=1$. Protocols with multi-qubit parties, require the use of quantum operations acting on two (or more) qubits, which are significantly harder to implement in all the known technologies. For instance in quantum optics, whose 'flying' qubits are well suited for pairability protocols, one-qubit operations are easy to perform using off-the-shelf standard devices, whereas two-qubit operations, like those required by the protocols using multi-qubit parties, can only be performed probabilistically with a non-negligible probability of failure \cite{bartolucci2023fusion,Lee2023graphtheoretical,ghanbari2024optimization,lu2007experimental}. Bravyi et al.~provided some results in the setup where each party holds one single qubit, although arguably weaker than those obtained in the case where each party holds at least 10 qubits. Using Reed-Muller codes,  they were able to construct a $k$-pairable state of size  exponential in $k$, namely $n = 2^{3k}$,  leaving  the existence of a $k$-pairable states of size $n=poly(k)$ as an open problem. They also found a 2-pairable graph state of size $10$ and proved that there exists no stabilizer state on less than 10 qubits that is 2-pairable using LOCC protocol based on Pauli measurements.

A natural generalization is to consider quantum states satisfying a stronger property: for some integer $k$, it is possible to induce any stabilizer state on any subset of $k$ qubits, by means of LOCC protocols. We call these states $k$-stabilizer universal. Stabilizer states constitute a powerful resource for multipartite quantum protocols \cite{Hein06,MS08, javelle2012new, vrana2017entanglement}, and can be described, up to local \footnote{As we consider one qubit per party, "local" is to be understood as "on each single qubit independently".} unitaries, by the formalism of graph states: a subset of quantum states which are in one-to-one correspondence with (undirected, simple) graphs. $2k$-stabilizer universality is a stronger notion than $k$-pairability: any $2k$-stabilizer universal state is $k$-pairable, as EPR pairs are stabilizer states. 

Our contributions rely on the graph state formalism and the ability to characterize properties of quantum states using tools from graph theory. In particular, we reformulate $k$-pairability as a property of a graph (rather than a property of a quantum state), such that the graph state corresponding to a $k$-pairable graph, is $k$-pairable. Furthermore, we relate pairability to the standard notion of vertex-minor (a complete and up-to-date survey on vertex-minors can be found in \cite{kim2023vertex}). A graph $H$ is a vertex-minor of $G$ if one can transform $G$ into $H$ by means of local complementations\footnote{Local complementation on a vertex $u$ consists  in complementing the subgraph induced by its neighbors.} and vertex deletions. If $H$ is a vertex-minor of $G$ then the graph state $\ket{H}$ can be obtained from $\ket G$ using only local Clifford operations, local Pauli measurements and classical communications. Dahlberg, Helsen, and Wehner proved that the converse is also true when $H$ has no isolated vertices~\cite{DWH:transfo}. In~\cite{dahlberg2020transforming}, they proved that it is NP-complete\footnote{Where the size of the input is the number of qubits of the graph state.} to decide whether a graph state can be transformed into a set of EPR pairs on specific qubits using using only local Clifford operations, local Pauli measurements and classical communications. In~\cite{DHW:howtotransform}, they showed that it is also NP-complete to decide whether a graph state can be transformed into another one using using only local Clifford operations, local Pauli measurements and classical communications.

The graphical counterpart of $k$-stabilizer universal graph states are called $k$-vertex-minor universal graphs, introduced in~\cite{claudet2023small}: a graph is $k$-vertex minor universal if it has any graph defined on any $k$ of its vertices as a vertex minor. If a graph is $k$-vertex-minor universal then the corresponding graph state is $k$-stabilizer universal. 
Stabilizer universal states (and thus $k$-vertex-minor universal graphs) are useful in themselves beyond the fact that they imply pairability, as they can serve as a primitive for quantum protocols using multipartite entanglement.  For instance, in \cite{claudet2023small}, it is shown that stabilizer universal states constitute a resource to perform a robust pairability protocol, in the sense that it allows some known parties to be malicious, while ensuring the correctness of the protocol. Furthermore, the notion of stabilizer universality is stronger than the notion of pairability. Nevertheless, while previous work~\cite{claudet2023small} establishes the existence of $k$-stabilizer universal graph states of size $n = O(k^4 \ln(k))$ and  of $k$-pairable graph states of size $n = O(k^3 \ln^3(k))$, there are no known graph states that are $k$-pairable but not $2k$-stabilizer universal.

In this work, we provide both probabilistic and explicit constructions of $k$-stabilizer universal graph states resulting from $k$-vertex-minor universal graphs. While the results are interesting in themselves from a combinatorial perspective, they allow one to explicitly define quantum communication protocols: if a $k$-stabilizer universal graph state is prepared, and each qubit is sent to a different party, then, with the assumption that each party can perform local quantum operations and that they can share classical information, any stabilizer state on any $k$ qubits can be generated. Note that this includes any set of disjoint EPR pairs on less than $k$ qubits. The local operations to perform in order to induce a given subgraph state derive directly  from the proofs of our results. 

The main contributions of the paper are as follows. 
In the first part of the paper, we prove the existence of $k$-vertex-minor universal graphs of order $n=O(k^2)$, which is optimal according to~\cite{claudet2023small}. We adopt a probabilistic approach, exhibiting a family of random bipartite graphs of quadratic order in $k$, which are $k$-vertex-minor universal with  probability going to $1$ exponentially fast in $k$. On the practical side, in the proof we introduce an efficient algorithm that tries to generate any induced graph of order $k$ as a vertex-minor on any $k$ vertices of a random bipartite graph, and the proof yields a bound on the probability of failure of the algorithm.

The second part of the paper focuses on explicit constructions of $k$-vertex-minor universal graphs. We derive our constructions from the incidence graph of the projective plane over the finite field  $\mathbb{F}_q$, where $q$ is a prime power. We denote this graph by $\mathbb{G}_q$. It is a bipartite graph of order $n=2(q^2+q+1)$, with the same number ($n/2$) of left and right vertices, corresponding respectively to points and lines of the projective plane (equivalently,  $1$-dimensional and $2$-dimensional linear subspaces of $\mathbb{F}_q^3$).  We show  that $\mathbb{G}_q$ satisfies the $k$-vertex-minor universality property, with $k = \Theta(n^{1/4})$. Furthermore, we show that the  graph on the points of the projective plane, with edges connecting points corresponding to orthogonal $1$-dimensional linear subspaces of $\mathbb{F}_q^3$, is $k$-vertex-minor universal, again with $k = \Theta(n^{1/4})$. To the best of our knowledge, these are the first explicit constructions of $k$-vertex-minor universal graphs of order polynomial in $k$, significantly improving on the previous explicit construction of $k$-pairable states based on Reed-Muller codes from~\cite{bravyi2022generating}, with exponential overhead.

\section{Vertex-minor and stabilizer universality}

The goal of this section is to cover all notions gravitating around  $k$-pairability and $k$-vertex-minor universality properties. We  first define the above properties for graphs,  then we  discuss their implications on the corresponding graph states. 

We denote a graph as $G = (V,E)$, where $V$ is the vertex set and $E$ is the edge set. All graphs are assumed to be undirected and simple (without loops or multiple edges). A vertex subset $S\subseteq V$ is said to be \textbf{stable} if no two vertices in $S$ are adjacent. Bipartite graphs are denoted as $G=(L, R, E)$, where $V = L \sqcup R$, with $L$ and $R$ disjoint stable sets referred to as \textbf{left and right vertex sets}, respectively.  To avoid possible confusion, we may sometimes write $V(G), E(G), L(G)$, or $R(G)$. A \textbf{pairing} is a graph $G$ such that any vertex is incident to exactly one edge. Given a vertex $v\in V$ and a subset $U\subseteq V$, we denote by $N_G(v)$ the \textbf{neighborhood} of $v$ in $G$, consisting of vertices $v\in V$ adjacent to $v$.

A \textbf{local complementation} on a vertex $v$ of a graph $G$ consists in complementing the subgraph induced by the neighborhood of $v$, more precisely, it leads to a graph $G\star v$ such that $V(G\star v)=V(G)$ and $E(G\star v) = E(G)\oplus E(K_{N_G(v)})$ where $K_S$ denotes the complete graph on the vertices in $S$, and $\oplus$ denotes the symmetric difference of two sets. We say that $G'$ is a \textbf{vertex-minor} of $G$, if $G'$ can be obtained from $G$ by means of local complementations and vertex deletions.  Here we consider $V(G') \subseteq V(G)$ and require $G'$ to be obtained exactly (not up to an isomorphism of graphs), meaning that there exists a sequence of graph transformations consisting of local complementations and the deletions of the vertices of $V(G)\setminus V(G')$.

\begin{definition} \label{def:kvmu-kpair}
    Given a graph $G$, a vertex subset  $V' \subseteq V(G)$, and an integer $k > 0$, we say that:
    \begin{itemize}
        \item $G$  is $k$-\textbf{vertex-minor universal} on $V'$, if $k\ls |V'|$ and any  graph on any $k$ vertices in $V'$ is a vertex-minor of $G$. 
        \item $G$ is $k$-\textbf{pairable} on $V'$, if $k\ls |V'|/2$ and any pairing on any $2k$ vertices in $V'$ is a vertex-minor of $G$.
    \end{itemize}

    If any of the above properties is satisfied with $V' = V(G)$, we say that $G$ is $k$-\emph{vertex-minor universal} or that $G$ is $k$-\emph{pairable}, respectively.
\end{definition}

\begin{definition} \label{def:bipartite-kvmu-kpair}
    We say that a bipartite graph $G = (L, R, E)$ is \textbf{left} (resp. \textbf{right}) $k$-\textbf{vertex-minor universal} or $k$-\textbf{pairable} if the corresponding condition from \Cref{def:kvmu-kpair} is satisfied for $V' = L$ (resp. $V' = R$). We say that $G$ is \textbf{two-side} $k$-vertex-minor universal\,/\,$k$-pairable if it is both left and right $k$-vertex-minor universal\,/\,$k$-pairable.
\end{definition}

Graph states form a standard family of quantum states that can be represented using simple undirected graphs (Ref. \cite{Hein06} is an excellent introduction to graph states). Given a graph $G=(V,E)$, the corresponding \textbf{graph state} $\ket G$ is the $|V|$-qubit state: $$\ket G=\frac{1}{\sqrt 2^{|V|}}\sum_{x\in 2^V} (-1)^{|G[x]|}\ket x$$ where $|G[x]|$ is the size (number of edges) of the subgraph induced by $x$ $^($\footnote{With a slight abuse of notation we identify a subset (say $x=\{u_2,u_4\}$) of the set of qubits $V=\{u_1,\ldots ,u_5\}$ with its characteristic binary word  ($x=01010$).}$^)$. 

We shall alternatively refer to the vertex set $V$ as qubit set. A graph state $\ket G$ can be prepared as follows: initialize every qubit in  $\ket + = \frac{\ket 0+\ket 1}{\sqrt 2}$ then apply for each edge of the graph a $CZ$ gate on the corresponding pair of qubits, where  $CZ: \ket{ab}\mapsto (-1)^{ab}\ket{ab}$. The graph state $\ket G$ is the unique quantum state (up to a global phase) that, for every vertex $u\in V$, is a fixed point of the Pauli operator $X_uZ_{N_G(u)}$ $^($\footnote{It consists in applying $X:\ket a\mapsto \ket {1{-}a}$ on $u$ and $Z:\ket a \mapsto (-1)^a\ket a$ on each of its neighbors in $G$.}$^)$. Hence, graph states form a subfamily of stabilizer states. Formally, an $n$-qubit \textbf{stabilizer state} \cite{gottesman1998heisenberg} is a quantum state that is the simultaneous eigenvector with eigenvalue 1 of $n$ commuting and independent Pauli operators. A useful property is that any stabilizer state is related to some graph state by the application of local Clifford unitaries, and these unitaries can be computed efficiently \cite{VandenNest04}. For instance, an EPR pair is equal to $\ket{K_2}$ up to local Clifford unitaries, where $K_2$ is the graph with two vertices and one edge. Thus, it is equivalent to be able to generate any stabilizer state, or any graph state, on a given set of qubits, by means of LOCC protocols. We introduce below the notion of $k$-stabilizer universal states. 

\begin{definition}
A quantum state $\ket{\psi}$ is \textbf{$k$-stabilizer universal} (resp., \textbf{$k$-pairable}) if any stabilizer state on any $k$ qubits in $V$ (resp., any $k$ EPR pairs on any $2k$ qubits in $V$) can be induced by means of LOCC protocols. 
\end{definition}

If $H$ is a vertex-minor of a $G$ then the graph state $\ket{H}$ can be obtained from $\ket G$ using only local Clifford operations, local Pauli measurements and classical communications, and the converse is true when $H$ has no isolated vertices~\cite{DWH:transfo}. As a pairing on $2k$ vertices has no isolated vertices, we have the following:

\begin{proposition}
    A graph $G$ is $k$-pairable if and only if the corresponding graph state $\ket G$ is $k$-pairable using only local Clifford operations, local Pauli measurements, and classical communication.
\end{proposition}

In the case of vertex-minor universality and stabilizer universality, the characterization from \cite{DWH:transfo} does not apply directly, because of possible isolated vertices. For instance, $K_2$ is not $2$-vertex-minor universal since no local complementation can turn it into an empty graph. However, $\ket{K_2}$ is 2-stabilizer universal: with \emph{e.g.} an $X$-measurement on each qubit, one can map the corresponding graph state (a maximally entangled pair of qubits) to the graph state composed a tensor product of two qubits. To be able to state a characterization, a solution is to introduce \emph{destructive} measurements.

\begin{proposition}\label{prop:vm}
   Given two graphs $G$ and $H$ such that $V(H) \se V(G)$, $H$ is a vertex-minor of $G$ if and only if $\ket H$ can be obtained from $\ket G$ (on the qubits corresponding to $V(H)$) using only local Clifford operations, local \underline{destructive} Pauli measurements, and classical communications.    
\end{proposition}

\begin{proof} Notice that a similar statement -- involving non-destructive measurements and only valid when $H$ does not contain isolated vertices -- has been proved in \cite{DWH:transfo} (Theorem 2.2). We provide here a direct proof of \Cref{prop:vm} which is actually slightly simpler thanks to the use of destructive measurements. In the following proof all measurements are destructive. $(\Rightarrow)$ Local complementations can be implemented by means of local Clifford unitaries, and vertex deletions by means of $Z$-measurements together with classical communications and Pauli corrections \cite{VandenNest04}. $(\Leftarrow)$ We prove the property by induction on the number of measurements. If there are no measurements the property is true \cite{VandenNest04}. Otherwise, let $u$ be the first qubit to be measured. Assume $u$ is measured according to $P$ and $C_u$ is the Clifford operator applied on $u$ before the measurement. $C_u^\dagger P C_u$ is proportional to some Pauli operator $P_0\in \{X,Y,Z\}$: 
\\(i) if $P_0=Z$, then the measurement of $u$ can be interpreted as a vertex deletion and leads to $\ket{G\setminus u}$ up to Pauli corrections, so, by induction hypothesis, $H$ is a vertex minor of $G\setminus u$, so is of $G$. 
\\(ii) if $P_0=Y$, then the measurement of $u$ can be interpreted as a $Z$-measurement on $\ket{G\star u}$ (up to a Clifford operator on some other qubits), thus according to (i), $H$ is a vertex minor of $G\star u$, so is of $G$. 
\\(iii) if $P_0=X$ and $N_G(u)\neq \emptyset$ then  the measurement $u$ can be interpreted as a $Y$-measurement on $\ket{G\star v}$ with $v\in N_G(u)$ (up to local Clifford operations on qubits different from $u$), thus according to (ii) $H$ is a vertex minor of $G\star v$, so is of $G$. 
\\(iv) if $P_0=X$ and $N_G(u)=\emptyset$ then $\ket G = \ket {G\setminus u}\otimes \ket +_u$ so after the measurement of $u$ the state is $\ket {G\setminus u}$, thus, by induction hypothesis, $H$ is a vertex minor of $G\setminus u$, so is of $G$.  
\end{proof}

\begin{corollary}
    A graph $G$ is $k$-vertex-minor universal if and only if the corresponding graph state $\ket G$ is $k$-stabilizer universal using only local Clifford operations, local \underline{destructive} Pauli measurements, and classical communication.     
\end{corollary}

Relations between pairability, vertex-minor universality and stabilizer universality of graph and graph states, can be found in \Cref{fig:implications}. To the best of our knowledge, all known examples of $k$-stabilizer universal (resp. $k$-pairable) graph states come from $k$-vertex-minor universal (resp. $k$-pairable) graphs. Furthermore, to  date, it is not known whether there exist $k$-pairable states which are not $2k$-stabilizer universal. Throughout this paper, we will essentially focus on the existence and the explicit construction of $k$-vertex-minor universal graphs.

\begin{figure}[ht]
\centering
$$
\begin{array}
[c]{ccc}
\vspace{0.5em} G \text{ is $2k$-vertex-minor universal} & \Longrightarrow & \ket G \text{ is $2k$-stabilizer universal} \\
\vspace{0.5em} \Big\Downarrow &  & \Big\Downarrow \\
G \text{ is $k$-pairable} & \Longrightarrow & \ket G \text{ is $k$-pairable} 
\end{array}
$$
\caption{Implications between pairability, vertex-minor universality and stabilizer universality of graphs and graph states.}
\label{fig:implications}
\end{figure}

\section{Existence~of~$k$-vertex-minor~universal~graphs~of~order~quadratic~in~$k$}
\label{sec:existence-mvu-quadratic-order}

Given any $k$, a $k$-vertex-minor universal graph has at least a quadratic order in $k$:

\begin{proposition}[\!\cite{claudet2023small}]
    \!If a graph $G$ of order $n$ is $k$-vertex-minor universal then $k\!\!<\!\! \sqrt{2n\log_2(3)}\!+\!2$.
\end{proposition}

In this section we prove that this bound is tight asymptotically, \emph{i.e.} there exists $k$-vertex-minor universal graphs whose order grows quadratically with $k$. This greatly improves over the probabilistic construction of~\cite{claudet2023small}, where the existence of $k$-vertex-minor universal graphs of order $O(k^4\ln(k))$ was proven.

\begin{theorem}
    For any constant $\alpha>2$, there exists $k_0$ s.t.~for any $k>k_0$, there exists a $k$-vertex-minor universal graph $G$ of order at most $\alpha k^2$.
    \label{thm:existence_vmu}
\end{theorem}

The remaining of this section is a proof of \Cref{thm:existence_vmu}. First we bound the probability that some graph of order $k$ is not a vertex-minor of a random bipartite graph $G$, in \Cref{lemma:proba_full_rank}. Then we bound the probability that such a random bipartite graph is $k$-vertex-minor universal, in \Cref{prop:proba_vmu}, by defining some algorithm that tries to generate any graph as a vertex-minor of $G$. Finally, we prove that there exists a $k$-vertex-minor universal bipartite graph of quadratic order in $k$. More precisely, the probability of a random bipartite graph of quadratic order being $k$-vertex-minor universal goes to $1$ exponentially fast in $k$:

\begin{proposition}
    Fix constants $\epsilon > 0$, $c>2$, and $c'>  \frac{1+\epsilon}{\ln(2)}$. There exists $k_0$ s.t.~for any $k>k_0$, the random bipartite graph $G$ (the probability of an edge existing between two vertices, one in $L(G)$ and one in $R(G)$, is $1/2$, independently of the other edges) with $|L(G)|= \lfloor c'k\ln(k) \rfloor$ and $|R(G)| = \lfloor c k^2 \rfloor$, is $k$-vertex-minor universal with probability at least $1 - e^{-\epsilon k \ln(k)}$. 
    \label{prop:proba_exp_vmu}
\end{proposition}

\Cref{prop:proba_exp_vmu} will be proved alongside \Cref{thm:existence_vmu} in this section. Notation-wise, given a set $A$ and an integer $k$, ${A \choose k}$ refers to $\{B \se A ~|~|B|=k\}$.

\begin{lemma}
    \label{lemma:proba_full_rank}
    Consider a random bipartite graph $G$ with $|L(G)|\gs k$, $|R(G)| \gs 4{k \choose 2}+5$: the probability of an edge existing between two vertices (one in $L(G)$ and one in $R(G)$) is $1/2$, independently of the other edges. Take $k \in \mathbb{N}$ and consider a set of vertices $K \in {L(G) \choose k}$. The probability that there exists a graph defined on $K$ which is not a vertex-minor of $G$ is upper bounded by $e^{-\frac{\left(\frac{|R(G)|}{4}-{k \choose 2}+1\right)^2}{\left(\frac{7|R(G)|}{4}-{k \choose 2}+1\right)}}$.
\end{lemma}

\begin{proof}
    For some $j \in \mathbb N\setminus\{0\}$ and $X \in {R(G) \choose j}$, consider the incidence matrix $M_X$ of size $j \times {k \choose 2}$, whose column number $i$ represent the pairs of vertices of $K$ that are in the neighborhood of the $i^{th}$ vertex of $X$, in the sense that its entries are 1 if the pair of vertices $u$,$v$ is in its neighborhood, 0 else. Note that if there exists some $X \in {R(G) \choose {k \choose 2}}$ whose incidence matrix $M_X$ is of full column-rank, then any $2^{{k \choose 2}}$ graph defined on $K$ is a vertex-minor of $G$. Indeed, column number $i$ represents the edges (resp. non-edges) of $K$ to be toggled by a local complementation on the $i^{th}$ vertex of $X$. So now we will bound the probability of such a set $X$ existing within $R(G)$.

    For this purpose we will greedily try to construct the set $X \in {R(G) \choose {k \choose 2}}$, one vertex after the other, by considering each vertex in $R(G)$ one by one, and we will lower bound the probability of the event "there exists some $X \in {R(G) \choose {k \choose 2}}$ whose incidence matrix $M_X$ if of full column-rank" by the probability of success of the algorithm. The algorithm works as follows. Arbitrarily order the vertices of $R(G)$. At each step (say that we have $j$ vertices in $X$ at some step), suppose the corresponding matrix of incidence (of size $j \times {k \choose 2}$) full column-rank. We consider the next vertex $u \in R(G)$ in the list: if adding its corresponding vector to $M_X$ increases its column-rank, then we add $u$ to $X$, else we remove $u$ from the vertices to consider. The algorithm stops (and succeeds) if $M_X$ has ${k \choose 2}$ columns and is full column-rank. Let us show that the probability of a vertex $u$ increasing the column-rank of $M_X$ (if $j < {k \choose 2}$) is lower-bounded by $1/4$. 

    If $M_X$ is of rank $j < {k \choose 2}$, there exists a non-zero vector $W$ (\emph{i.e.} a set of pairs of vertices of $K$) which is orthogonal to all $j$ first vectors. $W$ can be seen as the characteristic function of the edges of some graph $H$ on the vertices of $L(G)$. Adding a vertex $u$ to $X$ increases the rank of $M_X$ if the vector $U$ of incidence of $u$ in $K$ is such that $U \cdot W = 1 \mod 2$ (because then $U$ is not in the span of $M_X$). Note that, if $H$ has exactly one edge, then there is exactly probability $\frac{1}{4}$ that $U \cdot W = 1 \mod 2$ (in this case the two ends of the unique edge of $H$ are connected to $u$, which happens with probability $\frac{1}{2} \times \frac{1}{2}$). As $H$ has at least one edge, it has at least one vertex of non-zero degree $z$. Let us draw randomly the neighborhood of $u$: first we draw among the vertices of $H\sm\{z\}$, then we add $z$ with probability $\frac{1}{2}$. The probability that an odd number of neighbors of $z$ are neighbors of $u$ is $1/2$, so drawing $z$ changes the parity of the number of edges in $H$ whose ends are both neighbors of $u$, with probability $1/2$. At the end of the day there is a probability of at least $\frac{1}{4}$ that $U \cdot W = 1 \mod 2$, so that $u$ increases the column-rank of $M_X$.
    
    Finally, the algorithm fails if we encounter more than $|R(G)|-{k \choose 2}+1$ vertices that did not increase the column-rank of $M_X$. Let us introduce a random variable $T$ that follows the distribution $B(|R(G)|,3/4)$. The probability that the algorithm fails is upper bounded by $\Pr(T \gs |R(G)|-{k \choose 2}+1)$. We'll use the Chernoff bound: With $\mu = \mathbb{E}[T] = \frac{3|R(G)|}{4}$, for any $\delta > 0$, $\Pr(T \gs (1+\delta)\mu) \ls e^{-\frac{\delta^2}{2+\delta}\mu}$. As we need $(1+\delta)\mu = |R(G)|-{k \choose 2}+1$, we take $\delta = \frac{|R(G)|-{k \choose 2}+1-\mu}{\mu} >0$ (thanks to the conditions on $|R(G)|$). The Chernoff bound then gives $$ \Pr\left(T \gs |R(G)|\!-\!{k \choose 2}\!+\!1\right) \ls e^{\!-\frac{\left(\frac{|R(G)|-{k \choose 2}+1-\mu}{\mu}\right)^2}{\left(\frac{|R(G)|-{k \choose 2}+1+\mu}{\mu}\right)}\mu} \!\!\!\!\!= e^{-\frac{\left(|R(G)|-{k \choose 2}+1-\mu\right)^2}{\left(|R(G)|-{k \choose 2}+1+\mu\right)}} = e^{-\frac{\left(\frac{|R(G)|}{4}-{k \choose 2}+1\right)^2}{\left(\frac{7|R(G)|}{4}-{k \choose 2}+1\right)}} $$

    So the probability of the existence of $X \se {R(G) \choose k}$ whose incidence matrix $M_X$ if of full column-rank is lower bounded by $ 1 -e^{-\frac{\left(\frac{|R(G)|}{4}-{k \choose 2}+1\right)^2}{\left(\frac{7|R(G)|}{4}-{k \choose 2}+1\right)}} $.

\end{proof}

\begin{lemma}
    \label{prop:proba_vmu}
    Consider a random bipartite graph $G$ with $|L(G)|\gs k$, $|R(G)|\gs4{k \choose 2}+5$: the probability of an edge existing between two vertices (one in $L(G)$ and one in $R(G)$) is $1/2$, independently of the other edges. The probability that $G$ is $k$-vertex-minor universal is lower bounded by $$  1 - \left( \frac{k}{2^{|L(G)|-k+1}} + e^{-\frac{\left(\frac{|R(G)|}{4}-{k \choose 2}+1\right)^2}{\left(\frac{7(|R(G)|-k)}{4}-{k \choose 2}+1\right)}}\right)\times {|L(G)|+|R(G)| \choose k} $$
\end{lemma}

\begin{proof}
    Given a set $K \in {V \choose k}$, we consider the bad event $A_K$: "there exists a graph defined on $K$ which is not a vertex-minor of $G$". The probability that $G$ is $k$-vertex-minor universal is, by definition, $\Pr\left( \bigcap_{K \in {V \choose k}} \overline{A_K}\right)$. Suppose that each probability $\Pr(A_K)$ is upper bounded by some $p$. Then, using the union bound, $$\Pr\left( \bigcap_{K \in {V \choose k}} \overline{A_K}\right) = 1 - \Pr\left(\bigcup_{K \in {V \choose k}} A_K \right) \gs 1 - \sum_{K \in {V \choose k}} \Pr(A_K) \gs 1 - p\times {|L(G)|+|R(G)| \choose k}$$

    Now, fix some $K \in {V \choose k}$. Let us upper bound $\Pr(A_K)$. For this purpose, let us show in the following how one can induce with high probability any graph on $K$ as a vertex-minor, by first turning the graph $G$ into a graph $G'$ defined on a subset of $V$, which is bipartite (as $G$), and such that the vertices of $K$ are all "on the left". Let $L_{\overline K}=L(G)\sm (K\cap L(G))$, $R_K=R(G)\cap K$. The algorithm works as follows. Roughly speaking, we use pivotings to move the vertices of $R_K$ from the right side to the left side. Given $a\in R(G)$ and $b\in L(G)$, pivoting an edge $ab$ in a bipartite graph $G$ produces the bipartite graph $G\wedge ab=G*a*b*a$ where the edges between $N_G(a)\setminus \{b\}$ and $N_G(b)\setminus \{a\}$ are toggled and vertices $a$ and $b$ are exchanged (in other words $N_{G\wedge ab}(b)=N_G(a)\Delta \{a,b\}$, $N_{G\wedge ab}(a)=N_G(b)\Delta \{a,b\}$, so that the graph is bipartite according to the partition $R':=R(G)\Delta\{a,b\}$, $L':=L(G)\Delta\{a,b\}$). Once all vertices of $R_K$ are moved to the left by means of pivotings, we then consider the induced subgraph which consists in removing the vertices that have been moved from the left to the right side. We obtain a bipartite graph such that all vertices of $K$ are on the left side, the idea is then to apply \Cref{lemma:proba_full_rank} to show that with high probability one can induce any graph on $K$ as a vertex-minor (using only local complementation on vertices on the right side). So we have to prove that the constructed graph behaves as a random bipartite graph: each edge exists independently with probability $1/2$.

    To this end we provide a little more details on the algorithm. Given the initial random bipartite graph $G$, we proceed as follows: given a vertex $a\in R_K$, if there is no edge between $a$ and $L_{\overline K}$, the algorithm fails. Otherwise, we consider an arbitrary vertex $b\in N_G(a)\cap L_{\overline K}$ and perform a pivoting on $ab$, then remove vertex $b$, leading to a graph $G\wedge ab\setminus b$ which is bipartite according to $R':=R\setminus \{a\}$, $L':=L(G)\Delta\{a,b\}$. We show in the following that this bipartite graph is random, \emph{i.e.} each edge exists independently with probability $1/2$.    
    \begin{itemize}
        \item For any $u\in R'$ we have $a\sim_{G\wedge ab} u$ if and only if $b\sim_{G} u$, so $\Pr(a\sim_{G\wedge ab} u)=\frac12$.
        \item For any $u\in R', v\in L'\setminus \{a\}$ we have $u\sim_{G\wedge ab} v$  if and only if  ($u\sim_{G} v$ XOR ($u\in N_G(b)\wedge v\in N_G(a)$)). As the event $u\in N_G(b)\wedge v\in N_G(a)$ is independent of the existence of an edge between $u$ and $v$ in $G$, we have $\Pr(u\sim_{G\wedge ab} v)=\frac12$.
    \end{itemize}
    Regarding independence, notice that the existence of an edge $(u,v)$ in $G\wedge ab\sm b$ is independent of the existence of all edges but $(u,v)$ in $G$. The independence of the existence of the edges in $G$ guarantees the independence in $G\wedge ab \sm b$.

    To sum up, starting from a random bipartite graph and a vertex $a\in R_K$, if there is an edge between $a$ and some vertex $b\in L_{\overline K}$, we move $a$ to the left side (by means of a pivoting) and remove $b$: the remaining graph is a random bipartite graph. The algorithm consists in repeating this process until $R_K$ is empty, leading to a random bipartite graph. If the procedure succeeds, we end up with a bipartite graph $G'$ with $|R(G')|=|R(G)|-|R_K|$, $|L(G')|=|L(G)|$ such that $R(G') \se R(G)$ and $L(G') \se L(G) \cup R_K$. Recall that we obtained $G'$ from $G$ using only local complementations and vertex-deletions, so any vertex-minor of $G'$ is a vertex-minor of $G$. At the end of the day, $\Pr(A_K)$ is upper bounded by the probability that $\textbf{(1)}$ the algorithm fails or $\textbf{(2)}$ there exists a graph defined on $K$ which is not a vertex-minor of $G'$.
    \begin{itemize}
        \item $\textbf{(1)}$: Let us upper bound the probability that the algorithm fails. The algorithm is composed of $|R_K|$ steps. At each step, the algorithm fails if there is no edge between $a$ and $L_{\overline K}$. The bipartite graph, at this point, is random: each edges between $a$ and the vertices of $L_{\overline K}$ exist independently with probability 1/2. The probability that there is no edge between $a$ and $L_{\overline K}$ is $\frac{1}{2^{|L_{\overline K}|}}$.  In general, $|L_{\overline K}|$ is lower bounded by $|L(G)|-k+1$, and $|R_K|$ is upper bounded by $k$. At the end of the day, using the union bound, the algorithm fails with probability at most $\frac{k}{2^{|L(G)|-k+1}}$. 
        \item $\textbf{(2)}$: Suppose the algorithm succeeds. We end up with a random bipartite graph $G'$. Using \Cref{lemma:proba_full_rank}, the probability that there exists a graph defined on $K$ which is not a vertex-minor of $G'$ is upper bounded by $$e^{-\frac{\left(\frac{|R(G)|-|R_K|}{4}-{k \choose 2}+1\right)^2}{\left(\frac{7(|R(G)|-|R_K|)}{4}-{k \choose 2}+1\right)}} \ls e^{-\frac{\left(\frac{|R(G)|}{4}-{k \choose 2}+1\right)^2}{\left(\frac{7(|R(G)|-k)}{4}-{k \choose 2}+1\right)}}$$
    \end{itemize}

    So, at the end of the day, $$ \Pr(A_K) \ls \frac{k}{2^{|L(G)|-k+1}} + e^{-\frac{\left(\frac{|R(G)|}{4}-{k \choose 2}+1\right)^2}{\left(\frac{7(|R(G)|-k)}{4}-{k \choose 2}+1\right)}}$$

    And then,
    $$\Pr\left( \bigcap_{K \in {V \choose k}} \overline{A_K}\right)  \ls 1 - \left( \frac{k}{2^{|L(G)|-k+1}} + e^{-\frac{\left(\frac{|R(G)|}{4}-{k \choose 2}+1\right)^2}{\left(\frac{7(|R(G)|-k)}{4}-{k \choose 2}+1\right)}}\right)\times {|L(G)|+|R(G)| \choose k}$$

\end{proof}

\begin{remark}
    \Cref{prop:proba_vmu} has concrete applications on its own right: in particular for any integer $k$, it yields an integer $n$ such that there exists a (bipartite) $k$-vertex-minor universal graph of order $n$. In general, one can infer a lower bound on the probability of generating a $k$-vertex-minor universal graph, for any choice of $k$ and $n$, using the algorithm presented in the proof of \Cref{prop:proba_vmu}. A table presenting orders for which some bipartite $k$-vertex-minor universal graph exists, as well as orders for with a randomly generated bipartite graph is $k$-vertex-minor universal with at least 99\% probability, for particular values of $k$ ranging from 3 to 100, can be found in \Cref{app:table}. Surprisingly enough, we observe that a small constant additive overheard in the order of the graph is sufficient to attain a high probability of generating a $k$-vertex-minor universal graph.  
\end{remark}

Now we are ready to conclude. Fix some constants $c>2$ and $c'>\frac{1}{\ln(2)}$. Let $G$ be a random bipartite graph $G$ with $|L(G)| = \lfloor c'k\ln(k) \rfloor$ and $|R(G)| = \lfloor ck^2 \rfloor$: the probability of an edge existing between two vertices (one in $L(G)$ and one in $R(G)$) is $1/2$, independently of the other edges.\\

Note $n = |V| = |L(G)| + |R(G)| = \lfloor c'k\ln(k) \rfloor + \lfloor ck^2 \rfloor$. Using \Cref{prop:proba_vmu}, the probability that $G$ is $k$-vertex-minor universal is lower bounded by $$  1 - \left( \frac{k}{2^{|L(G)|-k+1}} + e^{-\frac{\left(\frac{|R(G)|}{4}-{k \choose 2}+1\right)^2}{\left(\frac{7(|R(G)|-k)}{4}-{k \choose 2}+1\right)}}\right)\times {n \choose k} $$

Let us prove that this probability is positive with our choice of parameters, for some big enough $k$. It is sufficient to have: $$\textbf{(1)}~~  \frac{k}{2^{|L(G)|-k+1}}{n \choose k}< \dfrac{1}{2} \text{~~and~~} \textbf{(2)}~~ e^{-\frac{\left(\frac{|R(G)|}{4}-{k \choose 2}+1\right)^2}{\left(\frac{7(|R(G)|-k)}{4}-{k \choose 2}+1\right)}}{n \choose k}< \dfrac{1}{2} $$

Let us show that these equations are satisfied for any large enough $k$. Recall that ${n \choose k} \ls 2^{nH(k/n)}$ where $H(x) = -x\log_2(x)-(1-x)\log_2(1-x)$ is the binary entropy.

$\textbf{(1)}$: It is sufficient that $\log_2(k)+nH(k/n)-|L(G)|+k-1 < -1$. 

$\log_2(k)+nH(k/n)-|L(G)|+k-1  \sim_{k\to \infty} n \frac{k}{n}\log_2(\frac{k}{n}) - c'k\ln(k) = k (\log_2(k) - \log_2(n)) - c'k\ln(k) \sim_{k\to \infty} \frac{1}{\ln(2)}k\ln(k) - c'k\ln(k)$. The choice of $c'$ guarantees that for any large enough $k$, $\textbf{(1)}$ is satisfied.

$\textbf{(2)}$: It is sufficient that $nH(k/n)\ln(2)-\frac{\left(\frac{|R(G)|}{4}-{k \choose 2}+1\right)^2}{\left(\frac{7(|R(G)|-k)}{4}-{k \choose 2}+1\right)} < -\ln(2)$. $\frac{\left(\frac{|R(G)|}{4}-{k \choose 2}+1\right)^2}{\left(\frac{7(|R(G)|-k)}{4}-{k \choose 2}+1\right)}$  $\sim_{k\to \infty} \frac{\left(\frac{ck^2}{4}-\frac{k^2}{2}\right)^2}{\left(\frac{7ck^2}{4}-\frac{k^2}{2}\right)} = k^2 \dfrac{(c-2)^2}{4(7c-2)}$. We saw above that $nH(k/n)\ln(2) \sim_{k\to \infty} k\ln(k)$. The choice of $c$ guarantees that for any large enough $k$, $\textbf{(1)}$ is satisfied.\\

This proves that, for any large enough $k$, $G$ of order $\lfloor c'k\ln(k) \rfloor + \lfloor ck^2 \rfloor$, is $k$-vertex-minor universal with non-zero probability. Taking $\alpha > c$, for any large enough $k$, $\lfloor c'k\ln(k) \rfloor + \lfloor ck^2 \rfloor \ls \alpha k^2$, proving \Cref{thm:existence_vmu}.

Furthermore, we just saw that side \textbf{(1)} of the equation dominates \textbf{(2)} asymptotically. Thus, the probability of $G$ being $k$-vertex-minor universal is roughly lower bounded by $1 - 2^{\frac{1}{\ln(2)}k\ln(k) - c'k\ln(k)} = 1 - e^{-(ln(2)c'-1)k\ln(k)}$ as $k$ grows. Then, for any $\epsilon > 0$ such that $\epsilon < \ln(2)c'-1$, for any large enough $k$, $G$ of order $\lfloor c'k\ln(k) \rfloor + \lfloor ck^2 \rfloor$, is $k$-vertex-minor universal with probability at least $1 - e^{-\epsilon k \ln(k)}$, proving \Cref{prop:proba_exp_vmu}.

\section{Vertex-minor universal graphs from projective planes}
\label{sec:explicit-constructions-vmu-graphs}

In this section, we provide explicit constructions of families of $k$-vertex-minor universal graphs, of order $n$ proportional to $k^4$. Thus, the order of the constructed graphs scales  as the square of the asymptotically optimal graph order from \Cref{sec:existence-mvu-quadratic-order}. To the best of our knowledge, these are the first explicit constructions of $k$-vertex-minor universal (or $k$-pairable) graphs of order $n$ polynomial in $k$. 

We start in \Cref{subsec:lemmas-k-pairability-vmu} with some preparatory lemmas.  In \Cref{subsec:bipartite-graph-proj-plane}, we introduce a family of  \emph{bipartite incidence graphs} of projective planes, and study their  $k$-pairability and $k$-vertex-minor universality properties. In \Cref{subsec:reduced-graph} we introduce a new family of so-called \emph{reduced graphs} from projective planes, and investigate their $k$-vertex-minor universality properties. A brief comparison between the two constructions is provided in \Cref{subsec:comparison-bipartite-reduced}.

\subsection{Sufficient conditions for $k$-pairability and $k$-vertex-minor universality}
\label{subsec:lemmas-k-pairability-vmu}

Below and throughout \Cref{sec:explicit-constructions-vmu-graphs}, given a graph $G$, a vertex $v\in V(G)$, and a vertex subset $U\subseteq V(G)$, we shall use the \textbf{shorthand notation} $N_U(v) := N_G(v) \cap U$, that is,  \textbf{the set of neighbors of $v$ that belong to $U$} (in such a case, we shall always ensure that the context makes the choice of $G$ unambiguous). 

The following two lemmas give sufficient conditions for a bipartite graph $G$ to be one-side (\emph{i.e.}, left or right) $k$-pairable or $k$-vertex-minor universal. For simplicity, we state these conditions for the set of left vertices.  
	
\begin{lemma}\label{lemma:Pcondition}
    Let $G$ be a bipartite graph satisfying the following property:
    \begin{description}
        \item[{\rm\em (P)}] For any set of $2k$  vertices $K = \{u_1,v_1, u_2,v_2, \dots, u_k, v_k\} \subseteq L(G)$, there exist:
              \begin{enumerate}
                  \item  a set of $k$  vertices $C = \{c_1, c_2,\dots, c_k \} \subseteq L(G)$, with  $C\cap K = \emptyset$, and
                  \item  a  set of $2k$  vertices $S = \{\alpha_1,\beta_1, \alpha_2,\beta_2, \dots, \alpha_k, \beta_k\} \subseteq R(G)$, such that $N_{K \cup C} (\alpha_i) = \{u_i, c_i\}$ and $N_{K \cup C} (\beta_i) = \{v_i, c_i\}$, for all $i=1,\dots,k$.
              \end{enumerate}
    \end{description}
    Then $G$ is left $k$-pairable.
\end{lemma}
\begin{proof}
    We  use  first local complementation on vertices $\alpha_i$ and $\beta_i$ to create edges $(u_i, c_i)$ and $(v_i, c_i)$, followed by local complementation on vertices $c_i$ to create edges $(u_i, v_i)$, as desired. It is easily seen that no edges are created between $u_i$ and $K\setminus\{v_i\}$, or between $v_i$ and $K\setminus\{u_i\}$.
\end{proof}

\begin{lemma}\label{lemma:VMUcondition}
    Let $G$ be a bipartite graph satisfying the following property:
    \begin{description}
        \item[{\rm\em (VMU)}] For any set of $k$  vertices $K = \{u_1, u_2, \dots, u_k\} \subseteq L(G)$, there exist:
              \begin{enumerate}
                  \item a set of $k(k-1)/2$  vertices $C = \{c_{ij} \mid 1 \ls i < j \ls k \} \subseteq L(G)$, with  $C\cap K = \emptyset$, and
                  \item a  set of $k(k-1)$  vertices $S = \{\alpha_{ij},\beta_{ij} \mid 1 \ls i < j \ls k \} \subseteq R(G)$, such that 
                   $N_{K \cup C} (\alpha_{ij}) = \{u_i, c_{ij}\}$ and $N_{K \cup C} (\beta_{ij}) = \{u_j, c_{ij}\}$, for all $1\ls i < j \ls k$.
              \end{enumerate}
    \end{description}
    Then $G$ is left $k$-vertex-minor universal.
\end{lemma}
\begin{proof}  
(See also \Cref{fig:vmuIllustration}.)
The proof is similar to that of \Cref{lemma:Pcondition}. To create an edge between $u_i$ and $u_j$, we use  first local complementation on vertices $\alpha_{ij}$ and $\beta_{ij}$, followed by local complementation on vertex $c_{ij}$. This procedure does not create any other edge between the vertices of $K$. 
\end{proof}

Providing sufficient conditions for a bipartite graph $G$ to be $k$-vertex-minor universal (on the entire vertex set) is more involved. To induce an arbitrary graph with vertex set $K=K_1 \sqcup K_2$, where $K_1  \subseteq L(G)$ and $K_2\subseteq R(G)$, we may need to create  edges with both endpoints in either $K_1$ or  $K_2$, which can be dealt with by using conditions similar to those in \Cref{lemma:VMUcondition}, but also ``toggle'' (\emph{i.e.}, either create or remove, as needed)   edges between $K_1$ and $K_2$, which represents an additional difficulty. We give sufficient conditions for doing so, in the lemma below (see also \Cref{fig:vmuStarIllustration}). 

\begin{lemma}\label{lemma:PstarCondition}
    Let $G$ be a bipartite graph satisfying the following property:
    \begin{description}
        \item[{\rm\em (VMU$^\star$)}] For any set of $k$ vertices $K=K_1 \sqcup K_2$, with $K_1 = \{u_1, \ldots , u_{k_1}\} \subseteq L(G)$, and $K_2=\{\lambda_1, \ldots, \lambda_{k_2} \} \subseteq R(G)$, there exist:
    \begin{enumerate}
        \item a subset $C_1 = \{c_{ij} \mid 1 \ls i < j \ls k_1\} \subseteq L(G)$, such that  $C_1 \cap K_1 = \emptyset$ and $N_{K_2}(C_1) = \emptyset$,
        
        \item a subset $S_1 = \{\alpha_{ij}, \beta_{ij} \mid 1 \ls i < j \ls k_1\} \subseteq R(G)$, such that $S_1 \cap K_2 = \emptyset$ and \\ for all $ 1 \ls i < j \ls k_1, N_{K_1 \sqcup C_1}(\alpha_{ij}) = \{u_i, c_{ij}\}$ and $N_{K_1 \sqcup C_1}(\beta_{ij}) = \{u_j, c_{ij}\}$,
        
        \item a subset $\Omega  = \{\omega_{ij} \mid 1 \ls i \ls k_1, 1 \ls j \ls k_2\} \subseteq R(G)$ such that $\Omega \cap(K_2 \sqcup S_1) = \emptyset$ and \\ for all $1 \ls i \ls k_1, 1 \ls j \ls k_2, N_{K_1 \sqcup C_1}(\omega_{ij}) = \{u_i\}$,
        
        \item a subset $C_{2}=\{\gamma_{ij} \mid 1 \ls j \ls k_2, j < i \ls k_1 + k_2 \} \subseteq R(G)$ such that $C_2 \cap (K_2 \sqcup S_1 \sqcup \Omega) = \emptyset$ and $N_{K_1 \sqcup C_1}(C_2)=\emptyset$,
        \item a subset $S_2= \{a_{ij}, b_{ij} \mid 1 \ls j \ls k_2, j < i \ls k_1 + k_2  \} \subseteq L(G)$ such that $S_2 \cap (K_1 \sqcup C_1) = \emptyset$ and for all $1 \ls j \ls k_2$, $j < i \ls k_1 + k_2$, \\
        $\bullet\ N_{K_2 \sqcup S_1 \sqcup \Omega \sqcup C_2}(a_{ij}) = \{\lambda_i, \gamma_{ij}\}$ and $N_{K_2 \sqcup S_1 \sqcup \Omega \sqcup C_2}(b_{ij}) = \{\lambda_j, \gamma_{ij}\}$, \qquad\ \ \ if $i \ls k_2$ \\
        $\bullet\ N_{K_2 \sqcup S_1 \sqcup \Omega \sqcup C_2}(a_{ij}) = \{\omega_{(i-k_2)j}, \gamma_{ij}\}$ and  $N_{K_2 \sqcup S_1 \sqcup \Omega \sqcup C_2}(b_{ij}) = \{\lambda_j, \gamma_{ij}\}$, \ otherwise.
    \end{enumerate}
    \end{description}
    Then $G$ is $k$-vertex-minor universal.
\end{lemma}

\begin{proof}
      We start by removing all vertices that are not in any set defined in (VMU$^\star$). Then we proceed in the following three steps. 
      
      1) In case we need to create an edge $(u_i, u_j)$ for $1 \ls i < j \ls k_1$ between two vertices in $K_1$.  We first use local complementations on $\alpha_{ij}$ and $\beta_{ij}$ to create  edges $(u_i,c_{ij})$ and $(u_j, c_{ij})$ (no other edges are created) and then remove $\alpha_{ij}$ and $\beta_{ij}$. Then, we use local complementation on $c_{ij}$ to create the edge $(u_i, u_j)$ (no other edges are created). Finally, we remove vertex $c_{ij}$, thus only the edge $(u_i, u_j)$ has been constructed. 
      
    2)  In case we need to create an edge $(\lambda_i, \lambda_j)$ for $1 \ls j < i \ls k_2$ between two vertices in $K_2$. We first use local complementations on $a_{ij}$ and $b_{ij}$ to create edges $(\lambda_i,\gamma_{ij})$ and $(\lambda_j, \gamma_{ij})$ (no other edges are created) and then remove $a_{ij}$ and $b_{ij}$. Then, we use local complementation on $\gamma_{ij}$ to create the edge $(\lambda_i, \lambda_j)$ (no other edges are created). Finally, we remove vertex $\gamma_{ij}$, thus only the edge $(\lambda_i, \lambda_j)$ has been constructed. 
      
    3)  In case we need to toggle an edge $(u_i, \lambda_j)$ for $1 \ls i \ls k_1$ and $1\ls j \ls k_2$ between two vertices in $K_1$ and $K_2$. We first use local complementations on $a_{(i+k_2)j}$ and $b_{(i+k_2)j}$ to create edges $(\omega_{ij}, \gamma_{(i+k_2)j})$ and $(\lambda_j, \gamma_{(i+k_2)j})$ (no other edges are created) and then remove $a_{(i+k_2)j}$ and $b_{(i+k_2)j}$. Then, we use local complementation on $\gamma_{(i+k_2)j}$ to create the edge $(\omega_{ij}, \lambda_j)$ (no other edges are created). After that, we remove vertex $\gamma_{ij}$, thus only the edge $(\omega_{ij}, \lambda_j)$ has been constructed. Finally, we use local complementation on $\omega_{ij}$ to create the edge $(u_i, \lambda_j)$ (no other edges are created). Then, we remove vertex $\omega_{ij}$, thus only the edge $(u_i, \lambda_j)$ has been toggled.
\end{proof}

\begin{figure}[!t]
    \centering
    \begin{subfigure}{0.45\textwidth}
        \centering
        \begin{tikzpicture}[xscale=.6, yscale=.6]
            \draw[rounded corners=.5cm] (0, -1.5) rectangle (3, 1.5);
            \draw (.8, .8) node {$K$};
            \draw[rounded corners=.5cm] (5, -4) 
            rectangle (8, -1);
            \draw (7.2, -1.7) node {$S$};
            \draw[rounded corners=.5cm] (0, -5.5) rectangle (3, -2.5);
            \draw (.8, -3.2) node {$C$};
            \draw[rounded corners=.7cm] (-.2, 2.2) rectangle (3.2, -7.3);
            \draw[rounded corners=.7cm] (5-.2, 2.2) rectangle (8.2, -7.3);
            \draw (1.5, -8) node {$L(G)$};
            \draw (6.5, -8) node {$R(G)$};
            beginning of edges
            
            \draw (2.2, .1) -- (5.8, -2.5);
            \draw (2.2, -4) -- (5.8, -2.5);
            \draw (2.2, -.7) -- (5.8, -3.2);
            \draw (2.2, -4) -- (5.8, -3.2);

            \draw[fill=black] (2.2, -.7) circle (.04);
            \draw[fill=black] (2.2, .1) circle (.04);
            \draw (1.8, -.7) node {$u_j$};
            \draw (1.8, .1) node {$u_i$};
            
            \draw[fill=black] (2.2, -4) circle (.04);
            \draw (1.8, -4) node {$c_{ij}$};

            \draw[fill=black] (5.8, -2.5) circle (.04);
            \draw[fill=black] (5.8, -3.2) circle (.04);
            \draw (6.6, -2.5) node {$\alpha_{ij}$};
            \draw (6.6, -3.2) node {$\beta_{ij}$};
        \end{tikzpicture}
        \caption{(VMU) conditions.}
        \label{fig:vmuIllustration}
    \end{subfigure}\hfill%
    \begin{subfigure}{0.45\textwidth}
        \centering
    \begin{tikzpicture}[xscale=.6, yscale = .6]
            \draw[rounded corners=.4cm] (0, 0) rectangle (3, 2);
            \draw (.7, 1.3) node {$K_1$};
            \draw (1.8, .7) node {$u_j$};
            \draw (1.8, 1.3) node {$u_i$};
            \draw[fill=black] (2.2, .7) circle (.04);
            \draw[fill=black] (2.2, 1.3) circle (.04);
            
            \draw[rounded corners=.4cm] (5, -.5) rectangle (8, 1.5);
            \draw (7.3, .8) node {$K_2$};
            \draw (6.2, .5) node {$\lambda_j$};
            \draw[fill=black] (5.8, .5) circle (.04);
            
            \draw[rounded corners=.4cm] (5, -3+.3) rectangle (8, -1+.3);
            \draw (7.3, -1.4) node {$S_1$};
            \draw (6.2, -2+.3+.3) node {$\alpha_{ij}$}; 
            \draw (6.2, -2.3+.3) node {$\beta_{ij}$};
            \draw[fill=black] (5.7, -2+.6) circle (.04);
            \draw[fill=black] (5.7, -2) circle (.04);
            
            \draw[rounded corners=.4cm] (0, -3+.3) rectangle (3, -1+.3);
            \draw (.7, -1.4) node {$C_1$};
            \draw (1.8, -1.7) node {$c_{ij}$};
            \draw[fill=black] (2.3, -1.7) circle (.04);
            
            \draw[rounded corners=.4cm] (5, -5+.1) rectangle (8, -3+.1);
            \draw (7.3, -3.6) node {$\Omega$};
            \draw (6.2, -4.1) node {$\omega_{ij}$};
            \draw[fill=black] (5.7, -4.1) circle (.04);
            
            \draw[rounded corners=.4cm] (5, -3-4.1) rectangle (8, -1-4.1);
            \draw (7.3, -5.8) node {$C_2$};
            \draw (6.2, -6.1) node {$\gamma_{i'j}$};
            \draw[fill=black] (5.7, -6.1) circle (.04);
            
            \draw[rounded corners=.4cm] (0, -3-4.1) rectangle (3, -1-4.1);
            \draw (.7, -5.8) node {$S_2$};
            \draw (1.8, -6.1-.4) node {$a_{i'j}$};
            \draw (1.8, -6.1+.4) node {$b_{i'j}$};
            \draw[fill=black] (2.3, -6.1-.3) circle (.04);
            \draw[fill=black] (2.3, -6.1+.3) circle (.04);
            
            \draw[rounded corners=.6cm] (-.2, 2.2) rectangle (3.2, -7.3);
            \draw[rounded corners=.6cm] (5-.2, 2.2) rectangle (8.2, -7.3);
            \draw (1.5, -8) node {$L(G)$};
            \draw (6.5, -8) node {$R(G)$};
            beginning of edges
            
            \draw (2.2, 1.3) -- (5.7, -1.4);
            \draw (2.2, .7) -- (5.7, -2);
            \draw (2.3, -1.7) -- (5.7, -1.4);
            \draw (2.3, -1.7) -- (5.7, -2);

            \draw (2.2, 1.3) -- (5.7, -4.1);
            
            \draw(5.7, -6.1) -- (2.3, -6.4);
            \draw(5.7, -6.1) -- (2.3, -5.8);
            \draw(5.7, -4.1) -- (2.3, -6.4);
            \draw(5.8, .5) -- (2.3, -5.8);
        \end{tikzpicture}
        \caption{(VMU$^\star$) conditions where $i' := i + k_2$.}
        \label{fig:vmuStarIllustration}
    \end{subfigure}
    \caption{Illustration of the (VMU) and (VMU$^\star$) conditions from \Cref{lemma:VMUcondition} and \Cref{lemma:PstarCondition}.}
    \label{fig:lemmaIllustration}
\end{figure}

The following lemma is a generalization of \Cref{lemma:VMUcondition}  to the case of general (not necessarily bipartite) graphs.

\begin{lemma}\label{lemma:VMUcondition_circ}
    Let $G$ be a graph satisfying the following property:
    \begin{description}
        \item[{\rm\em (VMU$^\circ$)}] For any set of $k$  vertices $K = \{u_1, u_2, \dots, u_k\} \subseteq V(G)$, there exist:
              \begin{enumerate}
                  \item a set of $k(k-1)/2$  vertices $C = \{c_{ij} \mid 1 \ls i < j \ls k \} \subseteq V(G)$, such that $C$ is stable,  $C\cap K = \emptyset$, and $N_K(c_{ij}) = \emptyset$, for all $1\ls i < j \ls k$,
                        and
                  \item a set of $k(k-1)$  vertices $S = \{a_{ij},b_{ij} \mid 1 \ls i < j \ls k \} \subseteq V(G)$, such that $S$ is stable, $S\cap (K \cup C) = \emptyset$, $N_{K \cup C}(a_{ij}) = \{u_i, c_{ij}\}$ and $N_{K \cup C}(b_{ij}) = \{u_j, c_{ij}\}$, for all $1 \ls i < j \ls k$.
              \end{enumerate}
    \end{description}
    Then $G$ is $k$-vertex-minor universal.
\end{lemma}
\begin{proof}
    Whenever we need to create or to remove an edge between vertices $u_i, u_j\in K$, we use first local complementation on vertices $a_{ij}$ and $b_{ij}$ to create edges between $u_i$ and $c_{ij}$, and between $u_j$ and $c_{ij}$, and then we use local complementation on $c_{ij}$.
\end{proof}

\subsection{Bipartite graphs from projective planes}
\label{subsec:bipartite-graph-proj-plane}

Let $q > 0$ be a prime power, $\mathbb{F}_q$ be the finite field with $q$ elements, and $\mbox{PG}(2,q) := \left(\mathbb{F}_q^3\right)^{*}\!/\, \mathbb{F}_q^{*}$ be the \emph{projective plane} over $\mathbb{F}_q$. \emph{Points} and \emph{lines} of $\mbox{PG}(2,q)$ are identified respectively to $1$-dimensional and $2$-dimensional linear subspaces of $\mathbb{F}_q^3$. A line $\lambda$ passes through a point $a$ (we write $a\in \lambda$) if the 2-dimensional linear subspace of $\mathbb{F}_q^3$ corresponding to $\lambda$ contains the $1$-dimensional linear subspace corresponding to $a$. We will use the following properties of the projective plane:
\begin{itemize}
    \item $\mbox{PG}(2,q)$ has $q^2+q+1$ points and $q^2+q+1$  lines.
    \item Any line contains exactly $q+1$ points, and any point is contained in exactly $q+1$ lines.
    \item Any two distinct lines intersect in  one point, and for any two distinct points there is one unique line containing them.
\end{itemize}

 We denote by $\G$ the bipartite incidence graph of the projective plane $\mbox{PG}(2,q)$.
 Precisely, the set of left vertices  $L(\G)$ is the set of points of $\mbox{PG}(2,q)$, the set of right vertices $R(\G)$ is the set of lines of $\mbox{PG}(2,q)$, and the set of edges $E(\G)$ corresponds to incidences between points and lines, that is $E(\G) = \{(a,\lambda) \in L(\G)\times R(\G) \mid a \in \lambda\}$.

\begin{theorem} \label{theo:bipartite-kpair}
    Let $k$ be such that $k \ls (q+4)/5$. Then $\G$ is two-side $k$-pairable.
\end{theorem}

\begin{proof}
    Due to the symmetry of $\G$, it is enough to prove it is left $k$-pairable. For this, we will use  \Cref{lemma:Pcondition}. Let  $K = \{u_1,v_1, u_2,v_2, \dots, u_k, v_k\} \subseteq L(\G)$ be a set of $2k$  points. To construct the sets $C = \{c_1, c_2,\dots, c_k \} \subseteq L(\G)$ and  $S = \{\alpha_1,\beta_1, \alpha_2,\beta_2, \dots, \alpha_k, \beta_k\} \subseteq R(\G)$ from  the property (P)  in \Cref{lemma:Pcondition}, we will proceed by recursion. 

    \medskip \noindent First, since there are $q+1$ lines passing through $u_1$ and $|K\setminus\{u_1\}| = 2k-1 \ls q$, we may choose a line $\alpha_1$ passing through $u_1$ and  not passing through any other point in  $K\setminus\{u_1\}$. Similarly, let $\beta_1$ be a line passing through $v_1$ and  not passing through any other point in  $K\setminus\{v_1\}$. We take  $c_1$ to be the intersection point between $\alpha_1$ and $\beta_1$.

    \medskip \noindent For $1\ls j < k$, assume that we have constructed a set of $j$ points $C_j = \{c_1, \dots, c_j \} \subseteq L(\G)$ and a set of $2j$ lines $S_j = \{\alpha_1,\beta_1,  \dots, \alpha_j, \beta_j\} \subseteq R(\G)$, satisfying the following conditions:
    \begin{enumerate}
        \item $C_j \cap K = \emptyset$,
        \item $N_{K \cup C_j} (\alpha_i) = \{u_i, c_i\}$ and $N_{K \cup C_j} (\beta_i) = \{v_i, c_i\}$, for all $i=1,\dots,j$.
    \end{enumerate}

    \noindent To construct $\alpha_{j+1}, \beta_{j+1}$, and $c_{j+1}$, we proceed in the following steps (see also \Cref{fig:kpair_recursion}). 
    \begin{itemize}
        \item  \textbf{We take $\alpha_{j+1}$ to be any line passing through $u_{j+1}$ and not passing through any other point in $\left(K\setminus\{u_{j+1}\}\right) \cup C_{j}$.} \\ 
        This is possible since $|(K\setminus\{u_{j+1}\}) \cup C_{j}| = 2k-1 + j \ls 3k-2 \ls q$. Moreover,  $\alpha_{j+1} \not\in S_j$, since by construction no line in $S_j$ passes through $u_{j+1}$.
        We further denote by $I_{j+1} \subseteq L(\G)$ the set consisting of the intersection points between $\alpha_{j+1}$ and the $2j$ lines in $S_j$. Thus, $|I_{j+1}| \ls 2j$.
        \item \textbf{We take $\beta_{j+1}$ to be any line passing through $v_{j+1}$ and not passing through any other point in $(K\setminus\{v_{j+1}\}) \cup C_{j} \cup I_{j+1}$.} \\
        This is possible since $|(K\setminus\{v_{j+1}\}) \cup C_{j} \cup I_{j+1}| \ls 2k-1 + 3j \ls 5k-4 \ls q$. Clearly,  $\beta_{j+1} \not\in S_j \cup \{\alpha_{j+1} \}$, since no line in $S_j \cup \{\alpha_{j+1} \}$ passes through $v_{j+1}$.
        \item \textbf{We take $c_{j+1}$ to be the intersection point between  $\alpha_{j+1}$ and $\beta_{j+1}$.} \\ Clearly,  $c_{j+1} \not\in C_j$, since neither one of $\alpha_{j+1}$ nor $\beta_{j+1}$ passes through the points in $C_j$.
    \end{itemize}

    \noindent To complete our recursion, we need to prove:
    \begin{enumerate}
        \item $C_{j+1} \cap K = \emptyset$. We only have to prove that $c_{j+1} \not\in K$. This follows from the fact that each of $\alpha_{j+1}$ and $\beta_{j+1}$ passes through only one point in $K$, namely $u_{j+1}$ and $v_{j+1}$, respectively, and they are distinct.
        \item $N_{K \cup C_{j+1}} (\alpha_i) = \{u_i, c_i\}$ and $N_{K \cup C_{j+1}} (\beta_i) = \{v_i, c_i\}$, for all $i=1,\dots,j+1$.

        For $i=j+1$, the above equalities follow  by construction. Indeed, $\alpha_{j+1}$ passes through $u_{j+1}$ and $c_{j+1}$, but it does not pass through any other point in $(K\setminus\{u_{j+1}\}) \cup C_{j}$, and similarly, $\beta_{j+1}$ passes through $v_{j+1}$ and $c_{j+1}$, but it does not pass through any other point in $(K\setminus\{v_{j+1}\}) \cup C_{j}$.

        For $1\ls i \ls j$, we only need to prove that neither $\alpha_i$ nor $\beta_i$ passes through $c_{j+1}$. This follows from the fact that $\beta_{j+1}$ does not pass through any point of $I_{j+1}$. Indeed, assuming that $c_{j+1}$  belongs to either $\alpha_i$ or $\beta_i$, implies it belongs to $I_{j+1}$, the set of intersection points between $\alpha_{j+1}$ and the lines in $S_j$. This  contradicts the fact that $\beta_{j+1}$ does not pass through any point of $I_{j+1}$.
    \end{enumerate}

    \noindent By recursion, we can construct sets $C := C_k$ and $S := S_k$ satisfying the property (P) from \Cref{lemma:Pcondition}, and thus we conclude that $\G$ is left $k$-pairable.
\end{proof}

\begin{figure}[!t]
    \centering
\begin{subfigure}{0.33\textwidth}
\resizebox{\columnwidth}{!}{%
\begin{tikzpicture}[x=.4cm,y=.5cm]
   \node at (4,8.5) {\textbf{(1)}};
   \coordinate (u1) at (8,6);
   \coordinate (v1) at (8,4);
   \coordinate (u2) at (12,6);
   \coordinate (v2) at (12,4);
   \node [fill=blue, circle, inner sep=2pt, label=135:${\color{blue}u_1}$] at (u1) {};
   \node [fill=blue, circle, inner sep=2pt, label=45:${\color{blue}u_2}$] at (u2) {};
   \node [fill=blue, circle, inner sep=2pt, label=-135:${\color{blue}v_1}$] at (v1) {};
   \node [fill=blue, circle, inner sep=2pt, label=-45:${\color{blue}v_2}$] at (v2) {};
   \draw [name path=alpha1] ($(u1)-(6,2)$) -- node[right=10,pos=-0.05] {$\alpha_1$} ($(u1)+(9,3)$);
   \draw [name path=beta1]  ($(v1)+(-6,2)$) -- node[right=10,pos=-0.05] {$\beta_1$} ($(v1)+(9,-3)$);
   \path [name intersections={of=alpha1 and beta1,by=c1}];
   \node [fill=blue, circle, inner sep=2pt, label=90:${\color{blue}c_1}$] at (c1) {};
   \node at ($(u2)+(-2,3)$) {}; 
   \node at ($(u2)+(4,-6)$) {}; 
   \node at ($(v2)-(2,3)$) {};  
   \node at ($(v2)+(4,6)$) {};  
   \end{tikzpicture}
}
\end{subfigure}\hfill%
\begin{subfigure}{0.33\textwidth}
\resizebox{\columnwidth}{!}{%
\begin{tikzpicture}[x=.4cm,y=.5cm]
   \node at (4,8.5) {\textbf{(2)}};
   \coordinate (u1) at (8,6);
   \coordinate (v1) at (8,4);
   \coordinate (u2) at (12,6);
   \coordinate (v2) at (12,4);
   \node [fill=blue, circle, inner sep=2pt, label=135:${\color{blue}u_1}$] at (u1) {};
   \node [fill=blue, circle, inner sep=2pt, label=45:${\color{blue}u_2}$] at (u2) {};
   \node [fill=blue, circle, inner sep=2pt, label=-135:${\color{blue}v_1}$] at (v1) {};
   \node [fill=blue, circle, inner sep=2pt, label=-45:${\color{blue}v_2}$] at (v2) {};
   \draw [name path=alpha1] ($(u1)-(6,2)$) -- node[right=10,pos=-0.05] {$\alpha_1$} ($(u1)+(9,3)$);
   \draw [name path=beta1]  ($(v1)+(-6,2)$) -- node[right=10,pos=-0.05] {$\beta_1$} ($(v1)+(9,-3)$);
   \path [name intersections={of=alpha1 and beta1,by=c1}];
   \node [fill=blue, circle, inner sep=2pt, label=90:${\color{blue}c_1}$] at (c1) {};
   \draw [name path=alpha2] ($(u2)+(-2,3)$) -- node[right=1,pos=0.02] {$\alpha_2$} ($(u2)+(4,-6)$);
   \node at ($(v2)-(2,3)$) {};  
   \node at ($(v2)+(4,6)$) {};  
   \path [name intersections={of=alpha2 and alpha1,by=a}];
   \node [fill=gray, circle, inner sep=2pt, label=90:${\color{gray}a}$] at (a) {};
   \path [name intersections={of=alpha2 and beta1,by=b}];
   \node [fill=gray, circle, inner sep=2pt, label=45:${\color{gray}b}$] at (b) {};
\end{tikzpicture}
}
\end{subfigure}\hfill%
\begin{subfigure}{0.33\textwidth}
\resizebox{\columnwidth}{!}{%
\begin{tikzpicture}[x=.4cm,y=.5cm]
   \node at (4,8.5) {\textbf{(3)}};
   \coordinate (u1) at (8,6);
   \coordinate (v1) at (8,4);
   \coordinate (u2) at (12,6);
   \coordinate (v2) at (12,4);
   \node [fill=blue, circle, inner sep=2pt, label=135:${\color{blue}u_1}$] at (u1) {};
   \node [fill=blue, circle, inner sep=2pt, label=45:${\color{blue}u_2}$] at (u2) {};
   \node [fill=blue, circle, inner sep=2pt, label=-135:${\color{blue}v_1}$] at (v1) {};
   \node [fill=blue, circle, inner sep=2pt, label=-45:${\color{blue}v_2}$] at (v2) {};
   \draw [name path=alpha1] ($(u1)-(6,2)$) -- node[right=10,pos=-0.05] {$\alpha_1$} ($(u1)+(9,3)$);
   \draw [name path=beta1]  ($(v1)+(-6,2)$) -- node[right=10,pos=-0.05] {$\beta_1$} ($(v1)+(9,-3)$);
   \path [name intersections={of=alpha1 and beta1,by=c1}];
   \node [fill=blue, circle, inner sep=2pt, label=90:${\color{blue}c_1}$] at (c1) {};
   \draw [name path=alpha2] ($(u2)+(-2,3)$) -- node[right=1,pos=0.02] {$\alpha_2$} ($(u2)+(4,-6)$);
   \draw [name path=beta2]  ($(v2)-(2,3)$) -- node[right=1,pos=0.02] {$\beta_2$} ($(v2)+(4,6)$);
   \path [name intersections={of=alpha2 and beta2,by=c2}];
   \node [fill=blue, circle, inner sep=2pt, label=0:${\color{blue}c_2}$] at (c2) {};
   \path [name intersections={of=alpha2 and alpha1,by=a}];
   \node [fill=gray, circle, inner sep=2pt, label=90:${\color{gray}a}$] at (a) {};
   \path [name intersections={of=alpha2 and beta1,by=b}];
   \node [fill=gray, circle, inner sep=2pt, label=45:${\color{gray}b}$] at (b) {};
\end{tikzpicture}
}
\end{subfigure}
    \caption{Recursive construction of sets $C$ and $S$ in the proof of of \Cref{theo:bipartite-kpair}, for $k=2$. (1) We chose $\alpha_1$ any line passing through $u_1$, and not passing through $v_1, u_2, v_2$. Similarly, we choose $\beta_1$  passing through $v_1$, and not passing through $u_1, u_2, v_2$. We take $c_1$ the intersection point between $\alpha_1$ and $\beta_1$. (2) We chose $\alpha_2$ any line passing through $u_2$, and not passing through $u_1, v_1, c_1, v_2$. We determine the intersection points $a$ and $b$ of $\alpha_2$ with $\alpha_1$ and $\beta_1$. (3) We chose $\beta_2$ any line passing through $v_2$, and not passing through $u_1, v_1, c_1, u_2$, as well as $a,b$ (to avoid $\alpha_2$ and $\beta_2$ intersecting on these points). We take  $c_2$ the intersection point between $\alpha_2$ and $\beta_2$.}
    \label{fig:kpair_recursion}
\end{figure}

\begin{theorem}\label{theo:bipartite-kvmu}
    Let $k$ be such that $3k^2-k-8 \ls 2q$. Then $\G$ is two-side  $k$-vertex-minor universal.
\end{theorem}

\begin{proof}
    Due to the symmetry of $\G$, it is enough to prove it is left $k$-vertex-minor universal. We prove $\G$ satisfies the property (VMU) from \Cref{lemma:VMUcondition}. Let  $K = \{u_1, u_2, \dots, u_k\} \subseteq L(\G)$ be a set of $k$  points. To construct the sets
    $C = \{c_{ij} \mid 1 \ls i < j \ls k \} \subseteq L(\G)$ and
    $S = \{\alpha_{ij},\beta_{ij} \mid 1 \ls i < j \ls k \} \subseteq R(\G)$ from \Cref{lemma:VMUcondition} we will proceed again by recursion, by running through pairs $(u_i, u_j)$ in some particular order, say in lexicographical order with respect to indexes $(i,j)$.

    \medskip \noindent The recursion is similar to the one in the proof of \Cref{lemma:Pcondition}.  We construct recursively lines $\alpha_{ij}$ and $\beta_{ij}$, passing through $u_i$ and $u_j$, respectively, and take $c_{ij} = \alpha_{ij} \cap \beta_{ij}$. In the recursion, we  take $\alpha_{ij}$ to be any line passing through $u_{i}$ and not passing through any other point in $\left(K\setminus\{u_{i}\}\right) \cup C_{ij}$, where $C_{ij} := \{ c_{i'j'} \mid (i',j') < (i,j)\}$. Since $|\left(K\setminus\{u_{i}\}\right)\cup C_{ij}| \ls (k-1) + (k(k-1)/2-1) = \frac{1}{2}(k^2+k-4)$, such a choice of $\alpha_{ij}$ is possible if  $k^2+k-4 \ls 2q$.

    \medskip \noindent A stronger constraint on the value of $k$ comes from the choice of $\beta_{ij}$. Indeed, for $\beta_{ij}$ we take any line passing through $u_{j}$ and not passing through any other point in $(K\setminus\{u_{j}\}) \cup C_{ij} \cup I_{ij}$, where $I_{ij}$ is the set of intersection points between $\alpha_{ij}$ and the previously constructed lines $\alpha_{i'j'}$ and $\beta_{i'j'}$, with $(i',j') < (i,j)$. Since $|(K\setminus\{u_{j}\}) \cup C_{ij} \cup I_{ij}| \ls (k-1) + 3(k(k-1)/2-1) = \frac{1}{2}(3k^2-k-8)$, we conclude that such a choice of $\beta_{ij}$ is possible as long as $\frac{1}{2}(3k^2-k-8) \ls q$, as stated in the lemma.
\end{proof}

\begin{restatable}{theorem}{firstbipartitevmu}
    \label{theo:bipartite-vmu}
    Let $k$ be such that $7k^2 - 16 \ls 4q$. Then $\G$ is $k$-vertex-minor universal.
\end{restatable}

\begin{proof}[Sketch of proof] Here, we will prove that $\G$ satisfies the property (VMU$^{\star}$) from \Cref{lemma:PstarCondition}.
    A first step is to build $C_1$ and $S_1$, used to construct edges, as needed, between vertices in $K_1$. This  is done similarly as in the proof of \Cref{theo:bipartite-kvmu}, except that they have to be disjoint from and not neighbors of $K_2$. Then, we choose $\Omega$ disjoint from $K_2 \sqcup S_1$, and such that the neighborhood of $\omega_{ij}$ intersects $C_1 \sqcup K_1$ in $u_i$ only. Finally, the last step consists in building the sets $S_2$ and $C_2$, used to construct edges, as needed, between vertices in $K_2 \sqcup \Omega$. This is done again similarly as in the proof of \Cref{theo:bipartite-kvmu} (with the role of left and right vertices inverted), and by taking into account that $S_2$ and $C_2$ have to be disjoint from and not neighbors of some of the previously constructed sets. The detailed proof is given in \Cref{proof:bipartite-vmu}. 
\end{proof}

\subsection{Reduced graphs from projective planes}
\label{subsec:reduced-graph}
\begin{definition}
    Let $G$ be a bipartite graph and $\varphi: L(G) \rightarrow R(G)$. The $\varphi$-\textbf{reduction} of $G$ is the graph $G_\varphi$ such that:
    \begin{itemize}
        \item The vertex set of $G_\varphi$ is the left vertex set of $G$, that is $V(G_\varphi) = L(G)$,
        \item There is an edge between $a,b\in V(G_\varphi)$, if $a\neq b$ and either ($a$ and $\varphi(b)$) or ($b$ and $\varphi(a)$) are neighbors in $G$, that is,
              $$E(G_\varphi) = \{ (a,b) \mid a \neq b \text{ and } [\,(a,\varphi(b)) \in E(G) \text{ or } (b,\varphi(a)) \in E(G)\,]\, \} $$
    \end{itemize}
    The reduction is said to be \textbf{bijective} if $\varphi$ is bijective. It is said to be \textbf{symmetric} if $\varphi$ is such that $(a,\varphi(b)) \in E(G) \Leftrightarrow (b,\varphi(a)) \in E(G),  \forall a,b \in L(G)$.   
\end{definition}
We enforce the condition $a\neq b$ in the definition of $E(G_\varphi)$, in order to avoid loops in case $(a,\varphi(a)) \in E(G)$ for some $a\in L(G)$.

\medskip Let $G_\varphi$ be a bijective, symmetric reduction of $G$. For any vertex $a\in V(G_\varphi) = L(G)$, let $N_{G_\varphi}(a) \subseteq L(G)$ be the set of neighbors of $a$ in $G_\varphi$, and $N_{G}(a) \subseteq R(G)$ be the set of neighbors of $a$ in $G$. By definition, if $b\in N_{G_\varphi}(a)$ then $\varphi(b)\in N_{G}(a)$. The converse is also true, except if $\varphi(a)\in N_{G}(a)$, or equivalently, $(a,\varphi(a)) \in E(G)$. Hence, $N_{G_\varphi}(a)  = 
    \{b \mid  \varphi(b) \in N_{G}(a)\} \setminus \{a\}$, and therefore:
\begin{itemize}
    \item If $(a,\varphi(a)) \not\in E(G)$, the map $\varphi$ induces a bijection between $N_{G_\varphi}(a)$ and $N_{G}(a)$. In particular,  $|N_{G_\varphi}(a)| =  |N_{G}(a)|$.
    \item If $(a,\varphi(a)) \in E(G)$, the map  $\varphi$ induces a bijection between $N_{G_\varphi}(a)$ and $N_{G}\setminus \{\varphi(a)\}$. In particular, $|N_{G_\varphi}(a)| =  |N_{G}(a)|-1$.
\end{itemize}

In what follows, we take $\G$ to be the bipartite incidence graph of the projective plane $\mbox{PG}(2,q)$ from the previous section. Let $\varphi: L(\G) \rightarrow R(\G)$   be defined as follows. Recall that a vertex $a\in L(\G)$ (that is, a point of the projective plane) corresponds to a $1$-dimensional linear subspace of  $\mathbb{F}_q^3$, while a vertex $\lambda\in R(\G)$ (that is, a line of the projective plane) corresponds to a $2$-dimensional linear subspace of  $\mathbb{F}_q^3$. Hence, for $a\in L(\G)$, we define $\varphi(a) \in R(\G)$ as the projective line corresponding to the $2$-dimensional linear subspace orthogonal to $a$. Clearly, $\varphi$ is bijective. It is also symmetric, since $a\in \varphi(b) \Leftrightarrow$ ($a$ and $b$ are orthogonal $1$-dimensional linear subspaces) $\Leftrightarrow b\in \varphi(a)$. Note also that  $(a, \varphi(a))\in E(\G)$ if and only if  $a$ is self-orthogonal. 

Let $\Gphi$   be the bijective, symmetric reduction of $\G$ induced by $\varphi$. We will not use the explicit definition of $\varphi$, but only the fact it is bijective and symmetric. Note that $\Gphi$ is a graph with $q^2+q+1$ vertices, and vertex degree equal to either $q$ (vertices corresponding to self-orthogonal linear subspaces) or $q+1$ (other vertices). The diameter of $\Gphi$ is equal to $2$, and for any two non-adjacent vertices $a, b\in V(\Gphi)$, there is a unique path of length $2$ connecting them.

\begin{theorem}
\label{theo:reduced-vmu}
    Let $k$ be such that $5k^2 -k -10 \ls 2q$. Then $\Gphi$ is $k$-vertex-minor universal.
\end{theorem}
\begin{proof}
    We prove $\Gphi$ satisfies the property (VMU$^\circ$) from \Cref{lemma:VMUcondition_circ}. Consider a set of $k$  vertices  $K = \{u_1, u_2, \dots, u_k\} \subseteq V(\Gphi)$. We will alternately refer to $u_i$'s as \emph{points}, since they are points of a projective plane (recall that $V(\Gphi) = L(\G)$). To construct the sets  $C = \{c_{ij} \mid 1 \ls i < j \ls k \} \subseteq V(\Gphi)$ and
    $S = \{a_{ij},b_{ij} \mid 1 \ls i < j \ls k \} \subseteq V(\Gphi)$ from \Cref{lemma:VMUcondition_circ} we will proceed again by recursion, by running through pairs $(u_i, u_j)$ in some particular order, say in lexicographical order with respect to indexes $(i,j)$.

    \medskip\noindent Fix some indexes $(i,j)$, with $1\ls i < j \ls k$, and assume that we have constructed (all neighbor sets below are defined with respect to the reduced graph $\Gphi$):
    \begin{enumerate}
        \item a set of vertices $C_{ij} = \{c_{i' j'} \mid (i',j') < (i, j)  \} \subseteq V(\Gphi)$, such that $C_{ij}$ is stable,  $K\cap C_{ij} = \emptyset$, and $N_K(c_{i'j'}) = \emptyset$, for all $(i',j') < (i, j)$,
        \item a set of vertices  $S_{ij} = \{a_{i'j'},b_{i'j'} \mid (i',j') < (i, j) \} \subseteq V(\Gphi)$, 
              such that $S_{ij}$ is stable,  $S_{ij}\cap (K \cup C_{ij}) = \emptyset$,  $N_{K \cup C}(a_{i'j'}) = \{u_{i'}, c_{i'j'}\}$ and $N_{K \cup C}(b_{i'j'}) = \{u_{j'}, c_{i'j'}\}$, for all $(i',j') < (i, j)$.
    \end{enumerate}
    Note that at the first step of the recursion, \emph{i.e.}, for $(i,j) = (1,2)$, we have $C_{1,2} = \emptyset$ and $S_{1,2} = \emptyset$. To proceed with the recursion, we need to construct $c_{ij}, a_{ij}$, and $b_{ij}$, which is done as follows (here below, for a set of points $P\subseteq L(\G)$, we denote by $\varphi(P) \subseteq R(\G)$ the image of $P$ through $\varphi$).
    \begin{itemize}
        \item \textbf{We choose $\alpha_{ij}\in R(\G) \setminus (\varphi(K) \cup \varphi(C_{ij}))$ a line passing through $u_i$,  and not passing through any point in $(K \setminus \{u_i\}) \cup C_{ij} \cup S_{ij}$.} \\
        There are $q+1$ lines passing through $u_i$, and to choose $\alpha_{ij}$ as above we need to avoid:
        \begin{itemize}
            \item $k$ lines in $\varphi(K)$ and $|C_{ij}| \ls k(k-1)/2-1$ lines in $\varphi(C_{ij})$,
            \item at most $(k-1)$ lines passing through $u_i$ and some point in $K \setminus \{u_i\}$,
            \item at most $|C_{ij}|$ lines passing through $u_i$ and some point in $C_{ij}$,
            \item at most $|S_{ij}| \ls k(k-1)-2$ lines passing through $u_i$ and some point in $S_{ij}$.
        \end{itemize}
        In total, the number of lines to avoid is upper bounded by $|K| + |K \setminus \{u_i\}| + 2|C_{ij}| + |S_{ij}| \ls (2k-1) + 2(k(k-1)-2) = 2k^2 -5$. Thus, such a  choice of $\alpha_{ij}$ is possible as long as $2k^2 -5 \ls q$, which is ensured by the upper bound on the $k$ value from the lemma.

        Let $I_{ij}^{(K)}$ be the set of intersection points between $\alpha_{ij}$ and the lines in $\varphi(K)$, and $I_{ij}^{(C)}$ be the set of intersection points between $\alpha_{ij}$ and the lines in $\varphi(C_{ij})$. Since $\alpha_{ij} \not\in \varphi(K) \cup \varphi(C_{ij})$, it follows that it intersects each line in $\varphi(K)$ or $\varphi(C_{ij})$ in one point. Thus, $|I_{ij}^{(K)}| \ls |K| = k$ and $|I_{ij}^{(C)}| \ls |C_{ij}|$.

        \item \textbf{We choose $\beta_{ij}\in R(\G) \setminus \varphi(C_{ij})$ a line passing through $u_j$, and not passing through any point in $(K \setminus \{u_j\}) \cup C_{ij} \cup S_{ij} \cup I_{ij}^{(K)} \cup I_{ij}^{(C)} \cup \{ \varphi^{-1}(\alpha_{ij}) \}$.} \\
        There are $q+1$ lines passing through $u_j$, and to choose $\beta_{ij}$ as above we need to avoid at most  $|K \setminus \{u_j\}| + 2|C_{ij}| + |S_{ij}| + |I_{ij}^{(K)}| + |I_{ij}^{(C)}| + 1 \ls (k-1) + 2(k(k-1)-2) + k + \frac{1}{2}k(k-1) = \frac{1}{2}(5k^2 -k -10)$ lines.  Such a choice is possible as long as  $5k^2 -k -10 \ls 2q$, which is exactly the upper bound on the value of $k$ from the lemma.

        \item Finally, \textbf{we take $c_{ij}$ to be the intersection point between lines $\alpha_{ij}$ and $\beta_{ij}$, $a_{ij} = \varphi^{-1}(\alpha_{ij})$ and $b_{ij} = \varphi^{-1}(\beta_{ij})$}.
    \end{itemize}

     To complete the recursion, we need to prove that the sets $\bar{C}_{ij} := C_{ij} \cup\{c_{ij}\}$ and $\bar{S}_{ij} := S_{ij} \cup\{\alpha_{ij}, \beta_{ij}\}$ satisfy the conditions  $(i)$ and  $(ii)$, above.

    \smallskip \noindent For $(i)$, we need to prove:
    \begin{itemize}
        \item $c_{ij} \not\in C_{ij}\cup K$. This follows from the fact that the line $\alpha_{ij}$ does not pass through any point in $ C_{ij}\cup (K \setminus \{u_i\})$ (and clearly, $c_{ij} \neq u_i$, since $\beta_{ij}$ does not pass through $u_i$).
        \item $c_{ij} $ is not adjacent (in $\Gphi$) to any $c_{i'j'}\in C_{ij}$. Assume that $c_{ij} $ is  adjacent to some $c_{i'j'}\in C_{ij}$. It follows that $c_{ij} \in  \varphi(c_{i'j'})$, and thus $c_{ij}\in I_{ij}^{(C)}$ (as it belongs to both $\alpha_{ij}$ and $\varphi(c_{i'j'})$), contradicting the fact that $\beta_{ij}$ does not pass through any point in $I_{ij}^{(C)}$.
        \item $N_K(c_{ij})  = \emptyset$. Assume that $c_{ij} $ is  adjacent to some $a \in K$. It follows that $c_{ij} \in  \varphi(a)$, and thus $c_{ij}\in I_{ij}^{(K)}$ (as it belongs to both $\alpha_{ij}$ and $\varphi(a)$), contradicting the fact that $\beta_{ij}$ does not pass through any point in $I_{ij}^{(K)}$.
    \end{itemize}

    \smallskip \noindent For  $(ii)$, we need to prove:
    \begin{itemize}
        \item $c_{ij} \not\in S_{ij} \cup \{a_{ij}, b_{ij}\}$. First, $c_{ij}\not\in S_{ij}$, which follows from the fact that $\alpha_{ij}$ does not pass through any point in $S_{ij}$. Assume that $c_{ij} = a_{ij}$.  Then $u_i \in \alpha_{ij} = \varphi(c_{ij})$,  meaning that $u_i$ and $c_{ij}$ are adjacent in $\Gphi$, which is impossible, since we already proved that $N_K(c_{ij})  = \emptyset$. Assuming $c_{ij} = b_{ij}$ leads to a similar contradiction.
        \item $a_{ij} \not\in K\cup C_{ij}$ and $b_{ij} \not\in K\cup C_{ij}$. Assuming $a_{ij}\in K$, we get $c_{ij} \in \alpha_{ij} = \varphi(a_{ij})$, meaning that $a_{ij}$ and $c_{ij}$ are adjacent in $\Gphi$, and thus contradicting $N_K(c_{ij})  = \emptyset$. Assuming $a_{ij}\in C_{ij}$ implies $\alpha_{ij} = \varphi(a_{ij}) \in \varphi(C_{ij})$, contradicting the choice of $\alpha_{ij}$. Similarly, $b_{ij} \not\in K\cup C_{ij}$.

        \item $a_{ij}$ is not adjacent (in $\Gphi$) to any $a_{i'j'}$ or $b_{i'j'}\in S_{ij}$. Assume that $a_{ij}$ and $a_{i'j'}$ are adjacent. Then $a_{i'j'} \in \varphi(a_{ij}) = \alpha_{ij}$, contradicting the fact that $\alpha_{ij}$ does not pass through any point of $S_{ij}$. Assuming $a_{ij}$ is adjacent to $b_{i'j'}$ leads to a similar contradiction.

        \item $b_{ij}$ is not adjacent (in $\Gphi$) to any $a_{i'j'}$ or $b_{i'j'}\in S_{ij}$. As above.

        \item $a_{ij}$ and $b_{ij}$ are not adjacent (in $\Gphi$). Assuming they are adjacent, we get $a_{ij} \in \varphi(b_{ij}) = \beta_{ij}$, contradicting the fact that $\beta_{ij}$ does not pass through $\varphi^{-1}(\alpha_{ij}) = a_{ij}$.

        \item $N_{K \cup C}(a_{ij}) = \{u_{i}, c_{ij}\}$. Since $u_i, c_{ij} \in \alpha_{ij} = \varphi(a_{ij})$, we have $\{u_{i}, c_{ij}\} \subseteq N_{K \cup C}(a_{ij})$. Assume that $a_{ij}$ is incident to some vertex $a\in K \setminus \{u_i\}$. Then $a\in \varphi(a_{ij}) = \alpha_{ij}$, which is impossible, since $\alpha_{ij}$ does not pass through any point in $K \setminus \{u_i\}$. Assume that $a_{ij}$ is incident to some vertex $c_{i'j'}\in C_{ij}$. Then $c_{i'j'} \in \varphi(a_{ij}) = \alpha_{ij}$, contradicting the fact that $\alpha_{ij}$ does not pass through any point of $C_{ij}$.

        \item $N_{K \cup C}(b_{ij}) = \{u_{j}, c_{ij}\}$. As above.
    \end{itemize}
    
    We conclude that $\bar{C}_{ij} = C_{ij} \cup\{c_{ij}\}$ and $\bar{S}_{ij} = S_{ij} \cup\{a_{ij}, b_{ij}\}$ satisfy the conditions $(i)$ and $(ii)$, above. By recursion, we can construct sets $C := \bar{C}_{k-1,k}$ and $S := \bar{S}_{k-1,k}$ satisfying the property (VMU$^\circ$) from \Cref{lemma:VMUcondition_circ}, and thus we conclude that $\Gphi$ is  $k$-vertex-minor universal.
\end{proof}

\subsection{Comparison between bipartite  and  reduced graph constructions}
\label{subsec:comparison-bipartite-reduced}

We denote by $\lceil x \rceil_\text{p}$ the smallest prime power greater than or equal to a real number $x > 1$. For a given $k > 1$, let $q_{2} := \left\lceil \frac{7}{4} k^2 - 4\right\rceil_\text{p}$ and $q_{1} := \left\lceil \frac{5}{2} k^2 - \frac{1}{2} k - 5\right\rceil_\text{p}$, given by the inequalities in \Cref{theo:bipartite-vmu} (bipartite graph) and \Cref{theo:reduced-vmu} (reduced graph), respectively. It follows that $\G[q_2]$ is a $k$-vertex-minor universal of order $n_2 = 2(q_2^2+q_2+1) \sim \frac{49}{8}k^4$, while $\Gphi[q_1]$ is a $k$-vertex-minor universal of order $n_1 = q_1^2+q_1+1 \sim \frac{25}{4}k^4$  (where $\sim$ indicates asymptotic equivalence, as $k$ goes to infinity). Thus, asymptotically, the bipartite graph construction yields $k$-vertex-minor universal graphs of slightly lower order than the reduced graph construction. Another interesting property of the  bipartite graph  is that the corresponding graph state $\ket{\G[q_2]}$ is equivalent, up to local Clifford unitaries, to a Calderbank-Shor-Steane (CSS) state~\cite[Section~IV]{chen2004multi}. However,  to construct a desired graph on $k$-vertices of the bipartite-graph $\G[q_2]$, we need to follow \Cref{lemma:PstarCondition}, thus to construct the sets $C_1, C_2, S_1, S_2, \Omega$ therein, which is done by following the steps highlighted in bold in the proof of \Cref{theo:bipartite-vmu}. Note that this directly translates into a LOCC protocol to induce a desired graph state on $k$ qubits of the state $\ket{\G[q_2]}$, using  \Cref{prop:vm}. For the reduced graph  the corresponding protocol is simpler, as we only have to construct the sets $C, S$ from \Cref{lemma:VMUcondition_circ}, which is again  done by following the steps highlighted in bold in the proof of \Cref{theo:reduced-vmu}.

\section{Conclusion}

We showed the existence of $k$-vertex-minor universal graphs of order quadratic in $k$, which attain the optimum. This implies the existence of $k$-vertex-minor universal and thus $k$-pairable graph states with a quadratic number of qubits. Then, our study of the incidence graph of a finite projective plane exhibited two families of $k$-vertex-minor universal graphs of linear order in $k^4$. These two families being, to our knowledge, the first $k$-stabilizer universal quantum states, and so $k$-pairable quantum states, that can be constructed on a polynomial number of qubits in $k$.\\
This leaves open some questions for future work.
 \begin{itemize}
    \item The logical next step is the explicit, deterministic construction of an infinite family of $k$-vertex-minor universal graphs whose order is cubic, or even quadratic in $k$, asymptotically matching the order of the $k$-vertex-minor universal graphs which can be constructed in a probabilistic, non-deterministic way (although with arbitrarily high probability).
     \item Our probabilistic construction for $k$-vertex-minor universal graphs is asymptotically optimal. The graph states corresponding to $2k$-vertex-minor universal graphs are also $k$-pairable: however the only known lower bound on the size of $k$-pairable states (where one party holds only one qubit) is quasi-linear \cite{bravyi2022generating}. Does there exist $k$-pairable states with a quasi-linear number of qubits? 
     \item Even though $2k$-stabilizer universality is a stronger requirement than $k$-pairability, it is not clear whether there exist $k$-pairable states which are not $2k$-stabilizer universal. A similar question can be asked for graphs: it is not clear whether there exist $k$-pairable graphs on more than $2$ vertices which are not $2k$-vertex-minor universal.
     \item Stabilizer universality of graph states, when restricted to LOCC protocols using only local Clifford operations, local destructive Pauli measurements, and classical communication, is fully characterized by the combinatorial properties of the associated graph. Does it extend to stabilizer universality with arbitrary LOCC protocols? Does the stabilizer universality of graph states translate to an augmented version of $k$-vertex-minor universality that allows vertex disconnection? An analog question may be asked for pairability: is the underlying graph of a $k$-pairable graph state always $k$-pairable?
     \item We chose to focus on the ability to induce any stabilizer state on any subset of $k$ qubit of a quantum state. What about the ability to induce any quantum state on any subset of $k$ qubit? How does it relate to $k$-stabilizer universality?
     \item Bravyi et al.~presented a construction of $k$-pairable states with an asymptotically optimal number of parties, in the case where each party holds at least 10 qubits \cite{bravyi2022generating}. How does $k$-stabilizer universality evolve when considering quantum communication networks where each party holds more than one qubit? Note that the construction of Bravyi et al.~where each party holds at least 10 qubits does not translate well for $k$-stabilizer universality.
 \end{itemize}

 \section*{Acknowledgements}
This work is supported by the PEPR integrated project EPiQ ANR-22-PETQ-0007 part of Plan France 2030, by the STIC-AmSud project Qapla’ 21-STIC-10, by the QuantERA grant EQUIP ANR-22-QUA2-0005-01, and by the European projects NEASQC and HPCQS.

\bibliography{biblio}

\begin{thebibliography}{10}

\bibitem{bartolucci2023fusion}
Sara Bartolucci, Patrick Birchall, Hector Bombin, Hugo Cable, Chris Dawson,
  Mercedes Gimeno-Segovia, Eric Johnston, Konrad Kieling, Naomi Nickerson,
  Mihir Pant, et~al.
\newblock Fusion-based quantum computation.
\newblock {\em Nature Communications}, 14(1):912, 2023.
\newblock \href {https://arxiv.org/abs/2101.09310} {\path{arXiv:2101.09310}},
  \href {https://doi.org/10.1038/s41467-023-36493-1}
  {\path{doi:10.1038/s41467-023-36493-1}}.

\bibitem{bravyi2022generating}
Sergey Bravyi, Yash Sharma, Mario Szegedy, and Ronald de~Wolf.
\newblock Generating $ k $ {EPR}-pairs from an $n$-party resource state.
\newblock {\em Quantum Information Processing}, 2023.
\newblock \href {https://arxiv.org/abs/2211.06497} {\path{arXiv:2211.06497}}.

\bibitem{chen2004multi}
Kai Chen and Hoi-Kwong Lo.
\newblock Multi-partite quantum cryptographic protocols with noisy {GHZ}
  states.
\newblock {\em Quantum Information and Computation}, 7(8), Nov 2007.
\newblock \href {https://arxiv.org/abs/quant-ph/0404133}
  {\path{arXiv:quant-ph/0404133}}, \href {https://doi.org/10.26421/QIC7.8-1}
  {\path{doi:10.26421/QIC7.8-1}}.

\bibitem{christandl2023resource}
Matthias Christandl, Vladimir Lysikov, Vincent Steffan, Albert~H Werner, and
  Freek Witteveen.
\newblock The resource theory of tensor networks, 2023.
\newblock \href {https://arxiv.org/abs/2307.07394} {\path{arXiv:2307.07394}}.

\bibitem{claudet2023small}
Nathan Claudet, Mehdi Mhalla, and Simon Perdrix.
\newblock Small $k$-pairable states, 2023.
\newblock \href {https://arxiv.org/abs/2309.09956} {\path{arXiv:2309.09956}}.

\bibitem{Contreras_Tejada_2022}
Patricia Contreras-Tejada, Carlos Palazuelos, and Julio~I. de~Vicente.
\newblock Asymptotic survival of genuine multipartite entanglement in noisy
  quantum networks depends on the topology.
\newblock {\em Physical Review Letters}, 128(22), 2022.
\newblock \href {https://arxiv.org/abs/2106.04634} {\path{arXiv:2106.04634}},
  \href {https://doi.org/10.1103/physrevlett.128.220501}
  {\path{doi:10.1103/physrevlett.128.220501}}.

\bibitem{DHW:howtotransform}
Axel Dahlberg, Jonas Helsen, and Stephanie Wehner.
\newblock How to transform graph states using single-qubit operations:
  computational complexity and algorithms.
\newblock {\em Quantum Science and Technology}, 5(4):045016, Sep 2020.
\newblock \href {https://arxiv.org/abs/1805.05306} {\path{arXiv:1805.05306}},
  \href {https://doi.org/10.1088/2058-9565/aba763}
  {\path{doi:10.1088/2058-9565/aba763}}.

\bibitem{dahlberg2020transforming}
Axel Dahlberg, Jonas Helsen, and Stephanie Wehner.
\newblock Transforming graph states to {B}ell-pairs is {NP}-{C}omplete.
\newblock {\em {Quantum}}, 4:348, Oct 2020.
\newblock \href {https://arxiv.org/abs/1907.08019} {\path{arXiv:1907.08019}},
  \href {https://doi.org/10.22331/q-2020-10-22-348}
  {\path{doi:10.22331/q-2020-10-22-348}}.

\bibitem{DWH:transfo}
Axel Dahlberg and Stephanie Wehner.
\newblock Transforming graph states using single-qubit operations.
\newblock {\em Philosophical Transactions of the Royal Society A: Mathematical,
  Physical and Engineering Sciences}, 376(2123):20170325, 2018.
\newblock \href {https://arxiv.org/abs/1805.05305} {\path{arXiv:1805.05305}},
  \href {https://doi.org/10.1098/rsta.2017.0325}
  {\path{doi:10.1098/rsta.2017.0325}}.

\bibitem{VandenNest04}
Maarten~Van den Nest, Jeroen Dehaene, and Bart~De Moor.
\newblock Graphical description of the action of local clifford transformations
  on graph states.
\newblock {\em Physical Review A}, 69(2), Feb 2004.
\newblock \href {https://arxiv.org/abs/quant-ph/0308151}
  {\path{arXiv:quant-ph/0308151}}, \href
  {https://doi.org/10.1103/physreva.69.022316}
  {\path{doi:10.1103/physreva.69.022316}}.

\bibitem{DSL:multipoint}
Gang Du, Tao Shang, and Jian-wei Liu.
\newblock Quantum coordinated multi-point communication based on entanglement
  swapping.
\newblock {\em Quantum Information Processing}, 16, Mar 2017.
\newblock \href {https://doi.org/10.1007/s11128-017-1558-2}
  {\path{doi:10.1007/s11128-017-1558-2}}.

\bibitem{fischer2021distributing}
Alex Fischer and Don Towsley.
\newblock Distributing graph states across quantum networks.
\newblock In {\em IEEE International Conference on Quantum Computing and
  Engineering (QCE)}, pages 324--333, 2021.
\newblock \href {https://arxiv.org/abs/2009.10888} {\path{arXiv:2009.10888}},
  \href {https://doi.org/10.1109/QCE52317.2021.00049}
  {\path{doi:10.1109/QCE52317.2021.00049}}.

\bibitem{ghanbari2024optimization}
Sobhan Ghanbari, Jie Lin, Benjamin MacLellan, Luc Robichaud, Piotr Roztocki,
  and Hoi-Kwong Lo.
\newblock Optimization of deterministic photonic graph state generation via
  local operations, 2024.
\newblock \href {https://arxiv.org/abs/2401.00635} {\path{arXiv:2401.00635}}.

\bibitem{gottesman1998heisenberg}
Daniel Gottesman.
\newblock The heisenberg representation of quantum computers, 1998.
\newblock \href {https://arxiv.org/abs/quant-ph/9807006}
  {\path{arXiv:quant-ph/9807006}}.

\bibitem{hahn2019quantum}
Frederik Hahn, Anna Pappa, and Jens Eisert.
\newblock Quantum network routing and local complementation.
\newblock {\em npj Quantum Information}, 5(1):1--7, 2019.
\newblock \href {https://arxiv.org/abs/1805.04559} {\path{arXiv:1805.04559}},
  \href {https://doi.org/10.1038/s41534-019-0191-6}
  {\path{doi:10.1038/s41534-019-0191-6}}.

\bibitem{Hein06}
Marc Hein, Wolfgang D{\"u}r, Jens Eisert, Robert Raussendorf, Maarten Van~den
  Nest, and Hans~J. Briegel.
\newblock Entanglement in graph states and its applications, 2006.
\newblock \href {https://arxiv.org/abs/quant-ph/0602096}
  {\path{arXiv:quant-ph/0602096}}.

\bibitem{illianoetal}
Jessica Illiano, Michele Viscardi, Seid Koudia, Marcello Caleffi, and
  Angela~Sara Cacciapuoti.
\newblock Quantum internet: from medium access control to entanglement access
  control, 2022.
\newblock \href {https://arxiv.org/abs/2205.11923} {\path{arXiv:2205.11923}}.

\bibitem{javelle2012new}
J{\'e}r{\^o}me Javelle, Mehdi Mhalla, and Simon Perdrix.
\newblock New protocols and lower bounds for quantum secret sharing with graph
  states.
\newblock In {\em Conference on Quantum Computation, Communication, and
  Cryptography}, pages 1--12. Springer, 2012.
\newblock \href {https://arxiv.org/abs/1109.1487} {\path{arXiv:1109.1487}}.

\bibitem{kim2023vertex}
Donggyu Kim and Sang-il Oum.
\newblock Vertex-minors of graphs: A survey, Oct 2023.
\newblock URL:
  \url{https://dimag.ibs.re.kr/home/sangil/wp-content/uploads/sites/2/2023/10/2023vertexminors-survey-revised.pdf}.

\bibitem{Lee2023graphtheoretical}
Seok-Hyung Lee and Hyunseok Jeong.
\newblock Graph-theoretical optimization of fusion-based graph state
  generation.
\newblock {\em {Quantum}}, 7:1212, Dec 2023.
\newblock \href {https://arxiv.org/abs/2304.11988} {\path{arXiv:2304.11988}},
  \href {https://doi.org/10.22331/q-2023-12-20-1212}
  {\path{doi:10.22331/q-2023-12-20-1212}}.

\bibitem{lu2007experimental}
Chao-Yang Lu, Xiao-Qi Zhou, Otfried G{\"u}hne, Wei-Bo Gao, Jin Zhang,
  Zhen-Sheng Yuan, Alexander Goebel, Tao Yang, and Jian-Wei Pan.
\newblock Experimental entanglement of six photons in graph states.
\newblock {\em Nature physics}, 3(2):91--95, 2007.

\bibitem{MS08}
Damian Markham and Barry~C. Sanders.
\newblock Graph states for quantum secret sharing.
\newblock {\em Physical Review A}, 78:042309, 2008.
\newblock \href {https://arxiv.org/abs/0808.1532} {\path{arXiv:0808.1532}},
  \href {https://doi.org/10.1103/PhysRevA.78.042309}
  {\path{doi:10.1103/PhysRevA.78.042309}}.

\bibitem{meignant2019distributing}
Cl\'ement Meignant, Damian Markham, and Fr\'ed\'eric Grosshans.
\newblock Distributing graph states over arbitrary quantum networks.
\newblock {\em Physical Review A}, 100:052333, Nov 2019.
\newblock \href {https://arxiv.org/abs/1811.05445} {\path{arXiv:1811.05445}},
  \href {https://doi.org/10.1103/PhysRevA.100.052333}
  {\path{doi:10.1103/PhysRevA.100.052333}}.

\bibitem{Miguel_Ramiro_2023}
Jorge Miguel-Ramiro, Alexander Pirker, and Wolfgang D{\"u}r.
\newblock Optimized quantum networks.
\newblock {\em Quantum}, 7:919, Feb 2023.
\newblock \href {https://arxiv.org/abs/2107.10275} {\path{arXiv:2107.10275}},
  \href {https://doi.org/10.22331/q-2023-02-09-919}
  {\path{doi:10.22331/q-2023-02-09-919}}.

\bibitem{pant2019routing}
Mihir Pant, Hari Krovi, Don Towsley, Leandros Tassiulas, Liang Jiang, Prithwish
  Basu, Dirk Englund, and Saikat Guha.
\newblock Routing entanglement in the quantum internet.
\newblock {\em npj Quantum Information}, 5(1):1--9, 2019.
\newblock \href {https://arxiv.org/abs/1708.07142} {\path{arXiv:1708.07142}},
  \href {https://doi.org/10.1038/s41534-019-0139-x}
  {\path{doi:10.1038/s41534-019-0139-x}}.

\bibitem{schoute2016shortcuts}
Eddie Schoute, Laura Mancinska, Tanvirul Islam, Iordanis Kerenidis, and
  Stephanie Wehner.
\newblock Shortcuts to quantum network routing, 2016.
\newblock \href {https://arxiv.org/abs/1610.05238} {\path{arXiv:1610.05238}}.

\bibitem{vrana2017entanglement}
P{\'e}ter Vrana and Matthias Christandl.
\newblock Entanglement distillation from {Greenberger--Horne--Zeilinger}
  shares.
\newblock {\em Communications in Mathematical Physics}, 352:621--627, 2017.
\newblock \href {https://arxiv.org/abs/1603.03964} {\path{arXiv:1603.03964}}.

\end{thebibliography}

\appendix

\section{Some data on the size of the existence constraints}

By \Cref{prop:proba_vmu}, given some $k\in \mathbb N\setminus\{0\}$, there exists a $k$-vertex-minor universal bipartite graph $G$ with $|L(G)|\gs k$, $|R(G)|\gs4{k \choose 2}+5$ if $$ \left( \frac{k}{2^{|L(G)|-k+1}} + e^{-\frac{\left(\frac{|R(G)|}{4}-{k \choose 2}+1\right)^2}{\left(\frac{7(|R(G)|-k)}{4}-{k \choose 2}+1\right)}}\right)\times {|L(G)|+|R(G)| \choose k} < 1$$
In \Cref{tab:exists_kvmu} we provide values for which there exists a $k$-vertex-minor universal bipartite graph of this order, for some particular values of $k$. In \Cref{tab:exists_kvmu99} we provide values for which a randomly generated bipartite graph is $k$-vertex-minor universal with at least 99\% probability. Experimentally, adding a small, constant number of vertices to the randomly generated bipartite graph, greatly increases the probability of it to be $k$-vertex-minor universal.

\label{app:table}

\begin{table}[ht]
    \centering
    \caption{Parameters for which some $k$-vertex-minor universal bipartite graph exists. }
    \begin{tabular}{|c|c|c|c|c|c|c|c|c|c|c|c|c|c|}
        \hline
        k               & 3  & 4  & 5  & 6  & 7  & 8  & 9  & 10  & 11  & 12  & 13  & 14  & 15  \\
        \hline
        |V(G)|  & 36  & 57  & 83  & 113  & 147  & 184  & 226  & 272  & 322  & 377  & 434  & 497  & 563 \\
        \hline
        $|L(G)|$  & 18  & 24  & 32  & 40  & 48  & 55  & 63  & 72  & 80  & 90  & 97  & 107  & 115 \\
        \hline
        $|R(G)|$ & 18  & 33  & 51  & 73  & 99  & 129  & 163  & 200  & 242  & 287  & 337  & 390  & 448 \\
        \hline
    \end{tabular}

    \bigskip

    \begin{tabular}{|c|c|c|c|c|c|c|c|c|c|c|c|c|c|}
        \hline
        k    & 20  & 25  & 30  & 35  & 40  & 50  & 60  & 70  & 80  & 90  & 100 \\
        \hline
        |V(G)|  & 955  & 1448  & 2041  & 2736  & 3531  & 5424  & 7718  & 10414  & 13512  & 17012  & 20912 \\
        \hline
        $|L(G)|$  & 161  & 208  & 256  & 306  & 357  & 461  & 568  & 677  & 788  & 902  & 1016 \\
        \hline
        $|R(G)|$  & 794  & 1240  & 1785  & 2430  & 3174  & 4963  & 7150  & 9737  & 12724  & 16110  & 19896 \\
        \hline
    \end{tabular}

    \label{tab:exists_kvmu}

\end{table}

\begin{table}[ht]
    \centering
    \caption{Parameters for which a randomly generated bipartite graph is $k$-vertex-minor universal with at least 99\% probability. }
    \begin{tabular}{|c|c|c|c|c|c|c|c|c|c|c|c|c|c|}
        \hline
        k               & 3  & 4  & 5  & 6  & 7  & 8  & 9  & 10  & 11  & 12  & 13  & 14  & 15  \\
        \hline
        |V(G)|  & 47  & 68  & 93  & 123  & 156  & 194  & 235  & 281  & 331  & 385  & 443  & 505  & 571 \\
        \hline
        $|L(G)|$  & 25  & 32  & 39  & 47  & 55  & 63  & 71  & 79  & 88  & 96  & 105  & 113  & 122 \\
        \hline
        $|R(G)|$ & 22  & 36  & 54  & 76  & 101  & 131  & 164  & 202  & 243  & 289  & 338  & 392  & 449 \\
        \hline
    \end{tabular}

    \bigskip

    \begin{tabular}{|c|c|c|c|c|c|c|c|c|c|c|c|c|c|}
        \hline
        k    & 20  & 25  & 30  & 35  & 40  & 50  & 60  & 70  & 80  & 90  & 100 \\
        \hline
        |V(G)|  & 962  & 1456  & 2049  & 2743  & 3539  & 5431  & 7726  & 10422  & 13519  & 17019  & 20920 \\
        \hline
        $|L(G)|$  & 167  & 215  & 263  & 313  & 364  & 468  & 575  & 684  & 795  & 908  & 1023 \\
        \hline
        $|R(G)|$  & 795  & 1241  & 1786  & 2430  & 3175  & 4963  & 7151  & 9738  & 12724  & 16111  & 19897 \\
        \hline
    \end{tabular}

    \label{tab:exists_kvmu99}

\end{table}

\section{Complete proof of \Cref{theo:bipartite-vmu} }
    \label{proof:bipartite-vmu}
    \firstbipartitevmu*
    
    \begin{proof}
        We want to prove that $\G$ satisfies the property (VMU$^\star$) of \Cref{lemma:PstarCondition}. Let $K= \{u_1, \ldots, u_{k_1}, \lambda_1, \ldots, \lambda_{k_2}\}$ be a set of $k$ points with $k_1\ls k$, $\{u_1, \ldots, u_{k_1}\} \subseteq L(\G)$, and $\{\lambda_1, \ldots, \lambda_{k_2}\} \subseteq R(\G)$. The proof is made of 3 inductions to construct: (1) the subsets $C_1 \subseteq L(\G)$ and $S_1\subseteq R(\G) $, (2) the subset $\Omega \subseteq R(\G)$, and finally (3) the subsets $C_2 \subseteq R(\G)$ and $S_2\subseteq L(\G)$, all defined as in the property (VMU$^\star$) from \Cref{lemma:PstarCondition}.

        (1) We construct the subsets $C_1 = \{c_{ij}\mid 1 \ls i<j \ls k_1\} \subseteq L(\G)$ and $S_1 = \{\alpha_{ij}, \beta_{ij} \mid 1 \ls i<j \ls k_1 \}\subseteq R(\G)$. The first induction is done on pairs $(u_i, u_j)$ in the lexicographical order over the indices and with the next induction hypothesis:\\

        (IH$_1^{ij}$) The subset of points $C_1^{ij} := \{c_{i'j'}\mid (i',j')<(i, j)\} \subseteq L(\G)$ and the subset of lines $S_1^{ij} = \{\alpha_{i'j'}, \beta_{i'j'}\mid (i', j') < (i,j) \}\subseteq R(\G)$ are such that
        \begin{enumerate}[(i)]
            \item $C_1^{ij} \cap K_1 = \emptyset$,
            \item $S_1^{ij} \cap K_2 = \emptyset$,
            \item $N_{K_2}(C_1^{ij})= \emptyset$,
            \item $N_{K_1 \sqcup C_1^{ij}}(\alpha_{i'j'}) = \{u_{i'}, c_{i'j'}\}, N_{K_1 \sqcup C_1^{ij}}(\beta_{i'j'}) = \{u_{j'}, c_{i'j'}\}$, for all $(i', j')<(i,j)$.
        \end{enumerate}

        The initialisation is done for $i=1$ and $j=2$. The subsets $S_1^{12} = \emptyset$ and $C_1^{12}=\emptyset$ satisfy the induction hypotheses (IH$_1^{12}$). 
        
        Then, for each induction step $(i,j)$ with $1 \ls i < j \ls k$. We construct $\alpha_{ij}, \beta_{ij}$ and $c_{ij}$, as follows.
        \begin{itemize}
            \item \textbf{We choose a line $\alpha_{ij}$ that does not belong to $K_2$ and passes only through $u_i$ in $C_1^{ij}\sqcup K_1 $.}
            There are $q+1$ lines passing through $u_i$ (the degree $d^\circ$ of a vertex is constant above~$\G$), at most $k_2$ of them belong to $K_2$. Moreover, the points of $C_1^{ij} \sqcup K_1 \setminus \{u_i\}$, paired with $u_i$, define at most $k_1(k_1-1)/2-1+ k_1 - 1 $ lines to avoid. Such line exists if \\ $\underbrace{q+1}_{d^\circ}- \underbrace{k_2}_{|N_{K_2}(u_i)|\ls} \gs \underbrace{1}_{\alpha_{ij}} + \underbrace{k_1 - 1}_{|K_1 \setminus \{u_i\}|} + \underbrace{k_1(k_1 - 1) /2 - 1}_{|C_1^{ij}|\ls}$ so if $k_1^2+2k-k_1-4\ls 2q$. Clearly, the line $\alpha_{ij}$ does not belong to $S_1^{ij}$ otherwise it would pass through one point of $C_1^{ij}$. Then the choice of $\beta_{ij}$ is driven by  $c_{ij}$, its intersection point with $\alpha_{ij}$. Indeed, the point $c_{ij}$ can not belong to a line of $K_2$, to a line of $S_1^{ij}$, and to the subset $C_1^{ij}$, the two first constraints are summarized by avoiding the two next sets:
            Let $I_{S_1}^{ij} \subseteq L(\G)$ be the set of intersection points of $\alpha_{ij}$ with lines of $S_1^{ij}$, with cardinality $|I_{S_1}^{ij}| \ls |S_1^{ij}| \ls k_1(k_1-1) - 2$. 
            Let $I_{K_2}^{ij} \subseteq L(\G)$ be the set of intersection points of $\alpha_{ij}$ with lines of $K_2$, with cardinality $|I_{K_2}^{ij}| \ls |K_2| \ls k_2$. 
            \item Thus, \textbf{we choose a line $\beta_{ij}$ that does not belong to $K_2$ and passes only through $u_j$ in $ K_1 \sqcup C_1^{ij} \sqcup I_{S_1}^{ij} \sqcup I_{K_2}^{ij}$}. Such choice is possible if
                  \[\underbrace{q+1}_{d^\circ} - \underbrace{k_2}_{|N_{K_2}(u_j)|\ls} \gs \underbrace{1}_{\beta_{ij}} + \underbrace{k_1 - 1}_{|K_1 \setminus \{u_j\}|} + \underbrace{k_1(k_1 - 1) /2 - 1}_{|C_1^{ij}|\ls} + \underbrace{k_1(k_1 - 1) - 2}_{|I_{S_1}^{ij}|\ls} + \underbrace{k_2}_{|I_{K_2}^{ij}|\ls} \]
                  and so if $2q \gs 3k_1^2+4k-5k_1-8$. The line $\beta_{ij}$ does not belong to $S_1^{ij} \sqcup \{\alpha_{ij}\}$ otherwise it would pass through a point of $C_1^{ij}$. 
            \item Finally, \textbf{we take the point $c_{ij}$ to be intersection point of $\alpha_{ij}$ and $\beta_{ij}$}. The point $c_{ij}$ does not belong to $C_1^{ij}$ as neither one of $\alpha_{ij}$ nor $\beta_{ij}$ passes through the points in $C_j$.
        \end{itemize}
        To end the first recursion, we need to prove:
        \begin{enumerate}[(i)]
            \item $c_{ij} \notin C_1^{ij} \sqcup K_1$. It follows from the line $\alpha_{ij}$ not passing through any point in $C_1^{ij} \sqcup K_1 \setminus \{u_i\}$ and the line $\beta_{ij}$ not passing through $u_i$. 
            \item $\alpha_{ij},\beta_{ij} \notin S_1^{ij} \sqcup K_2$. It follows from the avoidance of the points of $C_1^{ij}$ during the choice of the lines $\alpha_{ij}$ and $\beta_{ij}$.
            \item $N_{K_2}(c_{ij}) = \emptyset $. The point $c_{ij}$ does not belong to a line of $K_2$ as points of $I_{K_2}^{ij}$ are avoided when $\beta_{ij}$ is chosen.
            \item  $N_{K_1 \sqcup (C_1^{ij}\sqcup \{c_{ij}\})}(\alpha_{i'j'}) = \{u_{i'}, c_{i'j'}\},\ N_{K_1 \sqcup (C_1^{ij}\sqcup \{c_{ij}\})}(\beta_{i'j'}) = \{u_{j'}, c_{i'j'}\},$ \\$\forall (i', j')\ls(i,j)$. First, the line $\alpha_{ij}$ (resp. $\beta_{ij}$) passes only through the points $u_i$ (resp. $u_j$) and $c_{ij}$ in  the set $K_1 \sqcup C_1^{ij}$. Then,  for $(i', j') < (i,j)$, the lines $\alpha_{i'j'}$ and $\beta_{i'j'}$ do not pass through the point $c_{ij}$ otherwise $c_{ij}$ would belong to the subset $I_{S_1}^{ij}$ which is avoided during the choice of the line $\beta_{ij}$. 
        \end{enumerate}
        By recursion, we constructed the subsets $C_1 := C_1^{k_1, k_1}$ and $S_1 := S_1^{k_1, k_1}$ satisfying the first two constraints of the property (VMU$^\star$) from \Cref{lemma:PstarCondition}.
        
        \bigskip (2) Next step, we construct the subset $\Omega = \{\omega_{ij}\mid 1 \ls i \ls k_1, 1 \ls j \ls k_2\}$. We want one vertex $\omega_{ij}$ per pair $(u_i, \lambda_j)$ between $K_1$ and $K_2$. This set is constructed inductively for each value of $1 \ls i \ls k_1$. \\

        (IH$_2^{i}$) The subset of lines $\Omega^i = \{\omega_{i'j}, i' < i, 1 \ls j \ls k_2\} \subseteq R(\G)$ is such that         
        \begin{enumerate}[(i)]
            \setcounter{enumi}{4}
            \item $ (K_2 \sqcup S_1)\cap \Omega^i= \emptyset$
            \item $N_{K_1 \sqcup C_1}(\omega_{i'j}) = \{u_{i'}\}$ for all $1 \ls i' < i$ and $1 \ls j \ls k_2$.
        \end{enumerate}
        
        The initialisation is done for $i=1$, $\Omega^{0} = \emptyset$ satisfies (IH$_2^{i}$). 
        
        Then, for $1 \ls i \ls k_1$. \textbf{We choose $k_2$ distinct vertices $\{\omega_{ij}\}_{1 \ls j \ls k_2}$ adjacent to the vertex $u_i$ (one for each vertex $\lambda_j$ of $K_2$) such that  $\omega_{ij} \notin K_2 \sqcup S_1 \sqcup \Omega^i$ and $N_{K_1 \sqcup C_1}(\omega_{ij}) = \{u_i\}$}. The vertices of $S_1$ don't add constraints in the equation as they are already taken into account by the avoidance of vertices in $C_1$. Thus, such choice is possible if
        $$\underbrace{q+1}_{d^\circ} - \underbrace{k_2}_{|N_{K_2}(u_i)|\ls} - \underbrace{0}_{|N_{\Omega^i}(u_i)|} \gs \underbrace{k_2}_{\forall j, \omega_{ij}} + \underbrace{k_1 - 1}_{|K_1 \setminus \{u_i\}|} + \underbrace{k_1(k_1-1)/2}_{|C_1|}$$ 
        and so if $2q \gs k_1^2+4k-3k_1-4$ (where $k_2$ has been replaced by $k-k_1$). By construction, $\omega_{ij}$ avoid all points $u_{i'}$ for $i' \neq i$ and thus is distinct of $\omega_{i'j'}$ for $1 \ls i' < i$ and $1 \ls j j' \ls k_2$.
        
        To end the induction, we need to prove:
        \begin{enumerate}[(i)]
            \setcounter{enumi}{4}
            \item $ \{\omega_{ij}, 1\ls j \ls k_2\} \cap (\Omega^i \sqcup K_2 \sqcup S_1) = \emptyset$. By construction.
            \item $N_{K_1 \sqcup C_1}(\omega_{i'j}) = \{u_{i'}\}$ for all $1 \ls i' \ls i$ and $1 \ls j \ls k_2$. First, by construction the vertex $\omega_{ij}$ has only the vertex $u_i$ as neighbor in $K_1 \sqcup C_1$. For $i' < i$, nothing has changed. 
        \end{enumerate}
        
        By induction, we have constructed the set $\Omega := \Omega^{k_1+1}$ that satisfies the third constraint of (VMU$^\star$) from \Cref{lemma:PstarCondition}.
        
        \bigskip The remaining part of the proof consists in doing a truncated right $k$-vertex-minor universality with constraints coming from the already defined subsets $K_1, K_2, S_1, C_1, $ and $\Omega$. It will be really similar to the first part of this proof up to a swapping of the role of points and lines, swapping which is possible because of the symmetry property of the incidence graph of the finite projective plane. 

        (3) Finally, we build the sets $C_2=\{\gamma_{i,j}, 1 \ls j \ls k_2 \mid j < i \ls k_1 + k_2\} \subseteq R(\G)$ of lines and $S_2 = \{a_{ij}, b_{ij}, 1 \ls j \ls k_2, j < i \ls k_1 + k_2 \} \subseteq L(\G)$ of points. Indeed, we only need to control a subset of the edges in $K_2 \sqcup \Omega$, the subset of edges with two ends in $K_2$ and the subset of edges $\{(\omega_{ij}, \lambda_j)\}_{1 \ls i \ls k_1, 1 \ls j \ls k_2}$ with one end in $K_2$ and the other in $\Omega$.
        Thus, the induction is done according to the lexicographical order on the indexes $(i,j)$ for $1 \ls j \ls k_2$ and $j < i \ls k_2 + k_1$ where $(i,j)$ is used to index the pair $(\lambda_i, \lambda_j)$ if $i \ls k_2$, otherwise $(\omega_{(i-k_2)j}, \lambda_j)$. In the induction, we denote the by $r_{ij}$ the vertex of $R(\G)$ that corresponds to $\lambda_i$ if $i \ls k_2$ otherwise $\omega_{(i-k_2)j}$. 
        The induction hypothesis is: \\

        (IH$_3^{ij}$) The set of lines $C_2^{ij} = \{\gamma_{i'j'}\mid (i',j')<(i,j)\} \subseteq R(\G)$ and the set of points $S_2^{ij} = \{a_{i'j'}, b_{i'j'}\mid (i',j')<(i,j)\} \subseteq L(\G)$ are such that
        \begin{enumerate}[(i)]
            \setcounter{enumi}{6}
            \item $ \left(K_2 \sqcup S_1 \sqcup \Omega\right) \cap C_2^{ij} = \emptyset$,
            \item $ N_{C_2^{ij}} (K_1 \sqcup C_1) = \emptyset$,
            \item $(K_1 \sqcup C_1) \cap S_2^{ij}  = \emptyset$.
            \item $N_{K_2 \sqcup S_1 \sqcup \Omega \sqcup C_2^{ij}}(a_{i'j'}) = \{r_{i'j'},\gamma_{i'j'}\}, \quad N_{K_2 \sqcup S_1 \sqcup \Omega \sqcup C_2^{ij}}(b_{i'j'}) = \{\lambda_{j'}, \gamma_{i'j'}\}$ \\ for $(i',j') < (i', j')$.
        \end{enumerate}
        The initialisation is done for $i=2$ and $j=1$. We have that 
        $C_2^{ij} = \emptyset$ and $S_2^{ij} = \emptyset$ satisfy (IH$_3^{21}$).
        We then construct $a_{ij}, b_{ij}$, and $\gamma_{ij}$. The steps are the same as for the first induction so less detailed.
        \begin{itemize}
            \item \textbf{We choose a point $a_{ij}$ not in $N_{K_1 \sqcup C_1}(r_{ij})$ and that does not belong to any other line than $r_{ij}$ in $K_2 \sqcup S_1 \sqcup \Omega \sqcup C_2^{ij}$}. Such point exists if \[\underbrace{q+1}_{d^\circ} - \underbrace{k_1}_{|N_{K_1}(r_{ij})|\ls} - \underbrace{0}_{|N_{C_1}(r_{ij})|} \gs \underbrace{1}_{a_{ij}} + \underbrace{k_2 + k_1k_2 - 1}_{|(K_2 \sqcup \Omega)\setminus \{r_{ij}\}|} + \underbrace{k_1(k_1-1)}_{|S_1|} + \underbrace{k_2(k_2-1)/2+k_1k_2 - 1}_{|C_2^{ij}|}\] and thus if $2q \gs k^2 - k_1^2 + 2kk_1 + k - k_1 - 4$. With similar arguments on the double role of $C_2^{ij}$, $a_{ij} \notin S_2^{ij} \sqcup K_1 \sqcup C_1$. Then, the choice of $b_{ij}$ is driven by some constraints on the line passing through $a_{ij}$ and $b_{ij}$. Indeed, the line $\gamma_{ij}$ can not belong  to the subset $I_{S_2}^{ij} \subseteq R(\G)$ of lines defined by $a_{ij}$ and each point of $S_2^{ij}$ and of cardinality $|I_{S_2}^{ij}| \ls k_2(k_2-1) + 2k_1k_2 - 2$. The line $\gamma_{ij}$ can not also belong to the subset $I_{K_1\sqcup C_1}^{ij} \subseteq R(\G)$ of lines defined by $a_{ij}$ with each point of $K_1 \sqcup C_1$ and of cardinality $|I_{K_1\sqcup C_1}^{ij}| \ls k_1 + k_1(k_1-1)/2$.
            \item Then, \textbf{we choose the point $b_{ij}$ that does not belong to $N_{K_1 \sqcup C_1}(\lambda_j)$ and that does not belong to any other line that $\lambda_j$ in $K_2 \sqcup S_1 \sqcup \Omega \sqcup C_2^{ij} \sqcup I_{S_2}^{ij} \sqcup I_{K_1\sqcup C_1}^{ij}$}. Such choice of point is possible if \begin{align*} \underbrace{q+1}_{d^\circ} - &\underbrace{k_1}_{|N_{K_1}(\lambda_j)|\ls} - \underbrace{0}_{|N_{C_1}(\lambda_j)|} \gs \underbrace{1}_{b_{ij}} + \underbrace{k_2 - 1}_{|K_2 \setminus \{\lambda_j\}|} + \underbrace{k_1k_2}_{|\Omega|} + \underbrace{k_1(k_1-1)}_{|S_1|}\\ &+ \underbrace{k_2(k_2-1)/2+k_1k_2 - 1}_{|C_2^{ij}|} + \underbrace{k_2(k_2-1) + 2k_1k_2 - 2}_{|I_{S_2}^{ij}| \ls} + \underbrace{k_1 + k_1(k_1-1)/2}_{|I_{K_1\sqcup C_1}^{ij}| \ls} \end{align*} and so if $2q \gs 3k^2 - 2k_1^2+2kk_1+2k_1-k-8$. With similar arguments as before,  $b_{ij}$ does not belong to $S_2^{ij} \sqcup \{a_{ij}\} \sqcup K_1 \sqcup C_1$.
            \item Finally, \textbf{we take $\gamma_{ij}$ to be the line passing through  $a_{ij}$ and $b_{ij}$}. It is then clear that, $\gamma_{ij}$ does not belong to $C_2^{ij}$.
        \end{itemize}
        To end this last induction, we need to prove:
        \begin{enumerate}[(i)]
            \setcounter{enumi}{6}
            \item$\gamma_{ij} \notin (C_2^{ij} \sqcup K_2 \sqcup S_1 \sqcup \Omega) = \emptyset$. It follows from the construction as the point $a_{ij}$ (resp. $b_{ij}$) does not belong to any other line that $r_{ij}$ (resp. $\lambda_j$) in $K_2 \sqcup S_1 \sqcup \Omega$ and $r_{ij}$ is always distinct from $\lambda_j$.
            \item $N_{K_1 \sqcup C_1}(\gamma_{ij}) = \emptyset $. It follows from $\gamma_{ij} \notin I_{K_1\sqcup C_1}^{ij}$ as the point $b_{ij}$ can not belong to lines of $I_{K_1\sqcup C_1}^{ij}$.
            \item $a_{ij}, b_{ij} \notin (S_2^{ij} \sqcup K_1 \sqcup C_1)$. It is implied by construction constraints.
            \item $N_{K_2 \sqcup S_1 \sqcup \Omega \sqcup C_2^{ij}\sqcup \gamma_{ij}}(a_{i'j'}) = \{r_{i'j'},\gamma_{i'j'}\}, \quad N_{K_2 \sqcup S_1 \sqcup \Omega \sqcup C_2^{ij}\sqcup \gamma_{ij}}(b_{i'j'}) = \{\lambda_{j'}, \gamma_{i'j'}\}$ \\ for $(i',j') < (i', j')$. First, by construction $a_{ij}$ (resp. $b_{ij}$) does not belong to any other lines that $r_{ij}$ (resp. $\lambda_j$) and $\gamma_{ij}$ in $K_2 \sqcup \Omega \sqcup S_1 \sqcup C_2$. For $(i', j') < (i, j)$, the points $a_{i'j'}$ and $b_{i'j'}$ do not belong to the lines $\gamma_{ij}$ otherwise $\gamma_{ij}$ would belong to the subset $I_{S_2}^{ij}$ which is avoided during the choice of the point $b_{ij}$.
        \end{enumerate}

        By induction, we have constructed $S_2 := S_2^{(k_2+k_1)k_2} + \{a_{(k_2+k_1)k_2}, b_{(k_2+k_1)k_2}\}$ and $C_2 := C_2^{(k_2+k_1)k_2} + \{\gamma_{(k_2+k_1)k_2}\}$ satisfying the last two constraints of (VMU$^\star$) from \Cref{lemma:PstarCondition}. Thus, if all bounds on $q$ are satisfied then all constraints from (VMU$^\star$) are satisfied and then \Cref{lemma:PstarCondition} directly gives us that $\G$ is $k$-vertex-minor universal. As \Cref{theo:bipartite-vmu} is stated only with $k$, the bounds on $q$ have to be respected for all value of $k_1$ between 0 and $k/2$ (or between $k/2$ and $k$). To get the constraints independent of $k_1$, some basic calculus give us:
        \begin{itemize}
            \item  $2q \gs k^2/4 + 3k/2 -4$ for $\alpha_{ij}$,
            \item $2q \gs 3k^2/4 + 3k/2 - 8$ for $\beta_{ij}$,
            \item $2q \gs k^2/4+5k/2-4$ for $\omega_{ij}$,
            \item $2q \gs 9k^2/4+k/2-4$ for $a_{ij}$,
            \item $2q \gs 7k^2/2-8$ for $b_{ij}$.
        \end{itemize}
        Finally, for $k \gs 2$, the strongest one is $2q \gs 7k^2/2-8$. So, $\G$ is $k$-vertex-minor universal if $4q \gs 7k^2-16$. 
    \end{proof}

\end{document}